\documentclass[11pt, leqno, letterpaper]{amsart}

\usepackage{graphicx}    

\usepackage[margin=1.3in]{geometry}

\usepackage{amsmath,amsthm,amsfonts,amssymb}

\newtheorem{theorem}{Theorem}

\newtheorem{corollary}[theorem]{Corollary}
\newtheorem{definition}[theorem]{Definition}

\newtheorem{lemma}[theorem]{Lemma}

\newtheorem{proposition}[theorem]{Proposition}

\newtheorem{remark}[theorem]{Remark}

\begin{document}

\title[Asian Options for the CEV Model]
{Short Maturity Asian Options for the CEV Model}

\author{Dan Pirjol}
\email{
dpirjol@gmail.com}

\author{Lingjiong Zhu}
\address
{Department of Mathematics\newline
\indent Florida State University\newline
\indent 1017 Academic Way\newline
\indent Tallahassee, FL-32306\newline
\indent United States of America}
\email{
zhu@math.fsu.edu}

\date{15 January 2017}

\subjclass[2010]{91G20,91G80,60F10}
\keywords{Asian options, short maturity, CEV model, large deviations, variational problem.}

\begin{abstract}
We present a rigorous study of the short maturity asymptotics 
for Asian options with continuous-time averaging, under the assumption that 
the underlying asset follows the Constant Elasticity of Variance (CEV) model. 
We present an analytical approximation for the Asian options prices which
has the appropriate short maturity asymptotics, and 
demonstrate good numerical agreement of the asymptotic results with the
results of Monte Carlo simulations and benchmark test cases
for option parameters relevant in practical applications.
\end{abstract}

\maketitle

\section{Introduction}

Asymptotics for option prices and implied volatility of European options
for the short maturity regime 
have been extensively studied in the literature, see e.g. 
\cite{BBF,GHLOW,GW,HW,CCLMN} for local volatility models,
\cite{FF2012,Tankov,AL2013,AFMZ} for the exponential L\'{e}vy 
models and \cite{BBFII,HL,FJL,Feng,FengII,Forde,Alos} for stochastic 
volatility models and \cite{Gao,MN2011} for model-free approaches.

Recently this asymptotic regime was also investigated for Asian options in \cite{ShortMatAsian}
under the assumption that the asset price follows a local volatility model. 
The paper \cite{ShortMatAsian} considered arithmetic averaging Asian options in continuous time
under the assumption that the asset price follows a local volatility model
\begin{equation}\label{LVdef}
dS_{t}=(r-q)S_{t}dt+\sigma(S_{t})S_{t}dW_{t},
\qquad S_{0}>0,
\end{equation}
where $W_{t}$ is a standard Brownian motion, $r\geq 0$ is the risk-free rate, 
$q\geq 0$ is the continuous dividend yield, $\sigma(\cdot)$ is the local 
volatility function. 
The local volatility function $\sigma(\cdot)$ was assumed to satisfy the
boundedness and Lipschitz conditions 
\begin{align}
&0<\underline{\sigma}\leq\sigma(\cdot)\leq\overline{\sigma}<\infty,\label{assumpI}
\\
&|\sigma(e^x)-\sigma(e^y)|\leq M|x-y|^{\alpha},\label{assumpII}
\end{align}
for some fixed $M,\alpha>0$ for any $x,y$ and $0<\underline{\sigma}<\overline{\sigma}<\infty$ are some fixed constants.

Under these assumptions, it is known from \cite{Varadhan67}
that the log-stock price $X_t := \log S_t$
satisfies a sample path large deviation principle on an appropriate
functional space. 
This result was used in \cite{ShortMatAsian} together with the
contraction principle, to derive large deviations for the time average 
of the diffusion $\frac{1}{T} \int_0^T S_t dt$,
and short maturity asymptotics for out-of-money (OTM) Asian options. 
Using call-put parity, the corresponding short maturity
asymptotics for in-the-money (ITM) Asian options can be obtained as well. Finally,
the short maturity asymptotics for at-the-money (ATM) Asian options has been
derived too.
The result in \cite{ShortMatAsian} covers in particular the Black-Scholes
case, and the explicit formulas are derived for the short maturity asymptotics
for the Black-Scholes case in \cite{ShortMatAsian}.

The assumptions (\ref{assumpI}), (\ref{assumpII})
are not satisfied by some of the models which are popular in financial practice.
One important model of this type is the Constant Elasticity of Variance (CEV) 
model \cite{Cox}, which is defined by the diffusion
\begin{equation}\label{CEVdef}
dS_t = (r-q) S_t dt + \sigma S_t^\beta dW_t\,,\quad S_0>0\,.
\end{equation}
This model is used for modeling the skew in equities and FX markets,
and allows the flexibility of calibrating to the ATM slope of the implied
volatility by choosing appropriately the exponent $\beta$. For $\beta<1$, 
the model reproduces the leverage effect
observed in many financial markets, which is manifested as a decreasing
volatility as the asset price increases. 
The result of this inverse relationship between the price and volatility
is the implied volatility skew. 
See \cite{LinetskyMendoza} for a 
survey of the mathematical properties of the CEV model
and also the pricing of vanilla options under the CEV model. 

In most practical applications the exponent $\beta$ is usually chosen
in the range $0< \beta \leq 1$. 
The case $\beta = \frac12$ corresponds to the square-root model
of Cox and Ross \cite{CR}, 
and is obtained as a particular case of the Feller process \cite{Feller,CIR}
\begin{equation}\label{Feller}
dx_t = ( b x_t + c) dt + \sqrt{2a x_t} dW_t \,,
\end{equation}
with $a=\frac{1}{2}\sigma^{2}$, $b = r-q$, $c=0$. 
The case of general $\beta$ can also be mapped to the diffusion process
(\ref{Feller}) by a change of variable.
The classification of the solutions of the process (\ref{Feller}) has been 
studied by Feller \cite{Feller}. This can be used to obtain the
corresponding properties of the CEV model (\ref{CEVdef}), which are
summarized by the following well-known result, see 
\cite{AA2000,LinetskyMendoza}:

(i) $0 < \beta < \frac12$. The process (\ref{CEVdef})
can be mapped to the diffusion (\ref{Feller})
with $0< c< a$. The fundamental solution of Fokker-Planck equation for the density of the diffusion (\ref{CEVdef})
is not unique. There are two independent fundamental solutions, and the 
problem is well-posed only if we add an additional boundary condition at $x=0$,
for example absorbing or reflecting boundary condition.

(ii) $\frac12 \leq \beta <1$. The process (\ref{CEVdef}) can be mapped to 
the diffusion (\ref{Feller}) with $c< 0$. The Fokker-Planck equation for the density of the diffusion
has a unique fundamental 
solution, of decreasing norm.

The model (\ref{CEVdef}) is a local volatility model of type (\ref{LVdef}) with 
a volatility function $\sigma(S_t) := \sigma S_t^{\beta-1}$. 
For $0 < \beta < 1$ this is not a bounded function. 
This implies the results of \cite{Varadhan67} cannot be directly applied to
this case.

The pricing of Asian options has been widely studied in the mathematical
finance literature. The pricing under the Black-Scholes model has been 
studied in \cite{GY,CS,DufresneLaguerre,Linetsky}, using a relation between 
the distributional property of the time-integral of the geometric Brownian 
motion and Bessel processes. See \cite{DufresneReview} for an overview, and 
\cite{FMW} for a comparison with alternative simulation methods, such as the 
Monte Carlo approach. 

The PDE approach \cite{RogersShi,Vecer,VecerXu} can be used to price
Asian options under a wide variety of models, using either a numerical
approach \cite{Vecer,VecerXu}, 
or to derive analytical approximation formulae
using asymptotic expansion methods. The paper \cite{FPP2013} used heat kernel 
expansion methods and 
developed approximate formulae expressed in terms of elementary functions for 
the density, the price and the Greeks of path dependent options of Asian style.
Asymptotic expansion leading to analytical approximations with error bounds
for Asian options have been obtained also using Malliavin calculus in 
\cite{Shiraya,GobetMiri}.

Asian options pricing under the CEV model with $\beta=\frac12$
has been studied in \cite{DufresneSqrt} and \cite{DN2006}. A detailed study 
under the $\beta=\frac12$ model both with discrete and continuous time 
averaging was presented in \cite{FMR}. The general case of the
CEV model was studied in \cite{FPP2013} using heat kernel expansion methods 
in the PDE approach \cite{RogersShi,Vecer,VecerXu}.
The paper \cite{FPP2013} presented detailed numerical tests of their method under
the CEV model, which show good convergence and stability of the expansion.

In this paper, we study the short maturity asymptotics for the price
of the Asian options under the assumption that the underlying asset price
follows the CEV model (\ref{CEVdef})\footnote{We note that the short maturity
asymptotics for vanilla options under the CEV model has been studied in the
literature using several approaches, see \cite{HW,HL,CCLMN}.}. 
We consider both the fixed strike and floating strike Asian options.  
Our main tool is the large deviations theory from probability theory.
For the theory and applications of large deviations, we refer to the book 
\cite{Dembo}.
Some basic definitions and results needed in this paper will be provided in 
the Appendix.

The case of the square-root model $\beta=\frac12 $ is special as the model is 
affine, and the moment generating function of the time integral 
$\int_0^T S_t dt$ can be found in closed form. Then the
application of the G\"artner-Ellis theorem gives the large deviations for the averaged 
time integral of the asset price.

For $\frac{1}{2}<\beta<1 $ we use a recent large deviations result due to
Baldi and Caramelino \cite{BC2011} for the 
CEV model to derive a variational problem for the rate function determining
the short maturity asymptotics of the Asian options. Large deviations 
for the square-root process $\beta=\frac12$ were studied in \cite{DRYZ}.
The variational problem is solved completely. We derive large and small-strike 
asymptotics for the rate function. 

Some of the methods proposed in the literature for pricing Asian options are
less efficient in the small maturity/volatility limit. This is a well-known
problem in many of the methods proposed for the Black-Scholes model 
\cite{DufresneReview}, but a similar phenomenon appears
also for the method of \cite{DN2006} in the square-root model, where the 
convergence
of the expansion is slower for small maturity/volatility. The short maturity
asymptotic expansion proposed in this paper complements the use of these
methods in a regime where their numerical efficiency is less than optimal.

The paper is organized as follows. 
In Section \ref{Sec:2}, we present asymptotics for out-of-the-money (OTM) 
Asian options in the square-root model $\beta=\frac12$. 
Section \ref{Sec:3} considers the case of the general CEV model with 
$\frac12 \leq \beta <1$. The asymptotics for OTM Asian options is given by 
the solution of a variational problem, which is solved in closed form. 
We also obtain the asymptotics for at-the-money (ATM) Asian options. 
Section~\ref{Sec:4} considers the asymptotics of Asian options with
floating strike. In Section \ref{Sec:5} we present an analytical approximation
for the Asian options prices which has the same short maturity asymptotics
as that obtained in Sections~\ref{Sec:2} and \ref{Sec:3}. 
This approximation is compared against benchmark results in the literature,
and good agreement is demonstrated for model and option parameters relevant for
practical applications. 
Finally, the background of large deviations theory and the proofs of the 
main results are given in the Appendix.

\subsection*{Notations and preliminaries}

The price of the Asian call and put options with maturity $T$ and strike $K$ 
with continuous time averaging are given by expectations in the risk-neutral
measure
\begin{align}
&C(T):=e^{-rT}\mathbb{E}\left[\left(\frac{1}{T}\int_{0}^{T}S_{t}dt-K\right)^{+}\right],
\\
&P(T):=e^{-rT}\mathbb{E}\left[\left(K-\frac{1}{T}\int_{0}^{T}S_{t}dt\right)^{+}\right],
\end{align}
where $C(T)$ and $P(T)$ emphasize the dependence on the maturity $T$.

We denote the expectation of the averaged asset price in the risk-neutral 
measure as
\begin{equation}\label{Adef}
A(T) := \frac{1}{T} \int_0^T \mathbb{E}[S_t] dt = 
S_0 \frac{1}{(r-q)T} \left(e^{(r-q)T}-1\right)\,,
\end{equation}
for $r-q\neq 0$ and $A(T):=S_{0}$ for $r-q=0$,
When $K>A(T)$, the call Asian option is out-of-the-money and 
$C(T)\rightarrow 0$ as $T\rightarrow 0$.
When $A(T)>K$, the put Asian option is out-of-the-money and 
$P(T)\rightarrow 0$ as $T\rightarrow 0$.

The prices of call and put Asian options are related by put-call parity as
\begin{equation}\label{PCparity}
C(K,T) - P(K,T) = e^{-rT} (A(T) - K)\,.
\end{equation}

As $T\rightarrow 0$, we have $A(T)=S_{0}+O(T)$. 
Therefore, for the small maturity regime, the call Asian option is 
out-of-the-money if and only if $K>S_{0}$ etc.
For the purposes of the short maturity limit, 
the call Asian option is said to be out-of-the-money (resp. in-the-money) if $K>S_{0}$ (resp. $K<S_{0}$), 
and the put Asian option is said to be out-of-the-money (resp. in-the-money) if $K<S_{0}$ (resp. $K>S_{0}$),
and finally they are said to be at-the-money if $K=S_{0}$.

\section{Short Maturity Asian Options in the Square-root Model}
\label{Sec:2}

We assume in this Section that the asset value $S_{t}$ follows a Square-root 
process:
\begin{equation}
dS_{t}=(r-q)S_{t}dt+\sigma\sqrt{S_{t}}dW_{t},
\end{equation}
with $S_{0}>0$ and $W_{t}$ is a standard Brownian motion starting at 
zero at time zero $W_0=0$. 

We have the following result.

\begin{theorem}\label{ThmOTCSqrt}
$\mathbb{P}(\frac{1}{T}\int_{0}^{T}S_{t}dt\in\cdot)$ satisfies a large 
deviation principle with rate function
\begin{equation}
\mathcal{I}(x,S_0)=\sup_{\theta\in\mathbb{R}}\left\{\theta x-\Lambda(\theta)\right\},
\end{equation}
where 
\begin{equation}
\Lambda(\theta):=\lim_{T\rightarrow 0}T\log\mathbb{E}
\left[e^{\frac{\theta}{T^{2}}\int_{0}^{T}S_{t}dt}\right]
=
\begin{cases}
\frac{\sqrt{2\theta}}{\sigma}
\tan\left(\frac{\sigma}{2}\sqrt{2\theta}\right)S_{0} &\text{if $0\leq\theta<\frac{\pi^{2}}{2\sigma^{2}}$}
\\
\frac{-\sqrt{-2\theta}}{\sigma}\tanh\left(\frac{\sigma}{2}\sqrt{-2\theta}\right)S_{0} &\text{if $\theta\leq 0$}
\\
+\infty &\text{otherwise}
\end{cases}.
\end{equation}
\end{theorem}

Indeed, the rate function in Theorem~\ref{ThmOTCSqrt} 
has a more explicit expression. Together with Theorem \ref{ThmOTCSqrt} and
Lemma \ref{lemma:1} that we prove in Section \ref{Sec:3}, 
we have the following result.

\begin{proposition}\label{prop:2}
Assume the square root model: $\beta=\frac{1}{2}$.

(i) For $K\leq S_0$, the put option is OTM, and $P(T)=e^{-\frac{1}{T}\mathcal{I}(K,S_0)+o(1/T)}$, 
as $T\rightarrow 0$, where
\begin{equation}\label{18}
\mathcal{I}(K,S_0) = \frac{S_0}{\sigma^2}
\frac{x^2}{\cosh^2(x)} \left( \frac{\sinh(2x)}{2x} - 1\right)\,,
\end{equation}
where $x$ is the solution of the equation
\begin{equation}\label{19}
\frac{1}{2\cosh^2(x)} 
\left( 1 + \frac{\sinh(2x)}{2x} \right) = \frac{K}{S_0}\,.
\end{equation}

(ii) For $K\geq S_0$, the call option is OTM, and $C(T)=e^{-\frac{1}{T}\mathcal{I}(K,S_0)+o(1/T)}$, 
as $T\rightarrow 0$, where
\begin{equation}\label{20}
\mathcal{I}(K,S_0) = \frac{S_0}{\sigma^2}
\frac{x^2}{\cos^2(x)} \left( 1-\frac{\sin(2x)}{2x}\right)\,,
\end{equation}
where $0 \leq x \leq \frac{\pi}{2}$ is given by the solution 
of the equation
\begin{equation}\label{21}
\frac{1}{2\cos^2(x)} 
\left( 1 + \frac{\sin(2x)}{2x} \right) = \frac{K}{S_0}\,.
\end{equation}
\end{proposition}

We can study also the small/large strike asymptotics of the rate function.

\begin{proposition}\label{prop:4}

(i) The large strike asymptotics
for the rate function of OTM Asian call options $K>S_0$ in the square-root
model $\beta=\frac12$ is
\begin{equation}\label{LargeKI}
\lim_{K\rightarrow\infty}\frac{\mathcal{I}(K,S_{0})}{K}=
\frac{\pi^{2}}{2\sigma^{2}} \,.
\end{equation}

(ii) The small strike $K\to 0$ asymptotics of the rate function for 
OTM Asian put options $K<S_0$ in the square-root model $\beta=\frac12$ is
\begin{equation}\label{rem4eq}
\mathcal{I}(K,S_{0})\sim\frac{S_{0}^{2}}{2\sigma^{2}K},
\qquad\text{as $K\rightarrow 0$}.
\end{equation} 
\end{proposition}

\subsection{Expansion of the rate function around the ATM point}

We give also the expansion of the rate function for Asian options in the
square-root model ($\beta=\frac12$) in power series of $x=\log(K/S_0)$. 
The first few terms are
\begin{equation}\label{SqrtTaylor}
\mathcal{I}(K,S_0) = \frac{S_0}{\sigma^2}
\left\{ \frac32 x^2 + \frac35 x^3 + \frac{271}{1400} x^4 + O(x^5) \right\}.
\end{equation}
This gives an approximation for the rate function around the ATM point $x=0$.

The rate function $\mathcal{I}(K,S_0)$ in the square-root model was
evaluated numerically using the expression in Proposition~\ref{prop:2}.
We show in Figure~\ref{Fig:RateFunction} the 
plot of this function vs. $K/S_0$ (left) and vs. $x = \log(K/S_0)$ (right).
We show also in the right plot the approximation of the rate function
obtained by keeping the first three terms in the series expansion 
(\ref{SqrtTaylor}), which gives a good approximation around the ATM point $x=0$.


\begin{figure}[b!]
    \centering
   \includegraphics[width=2.5in]{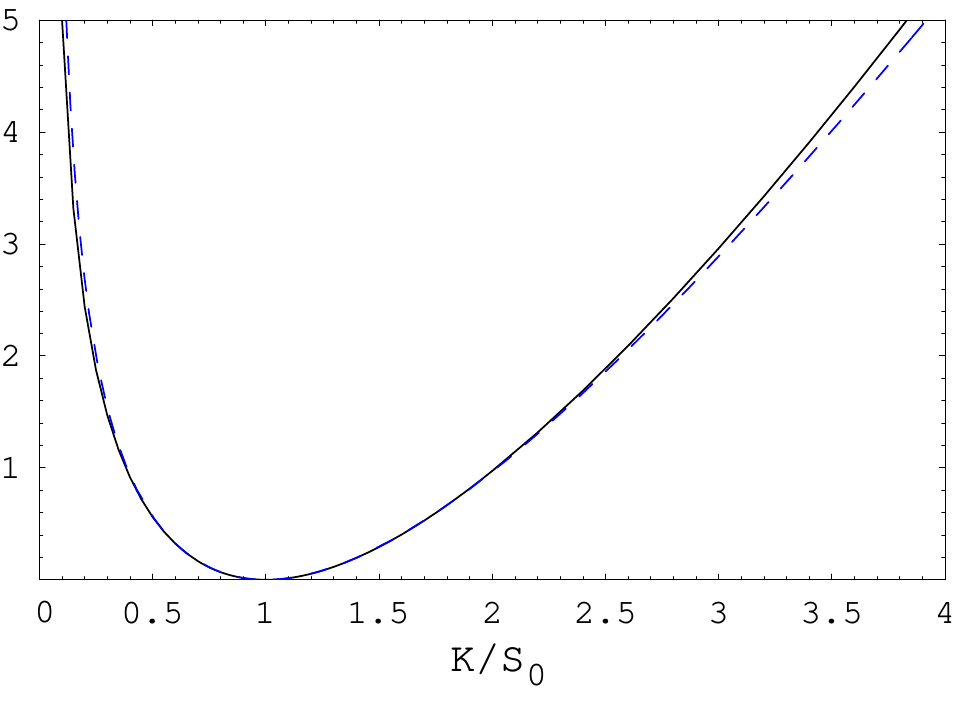}
   \includegraphics[width=2.5in]{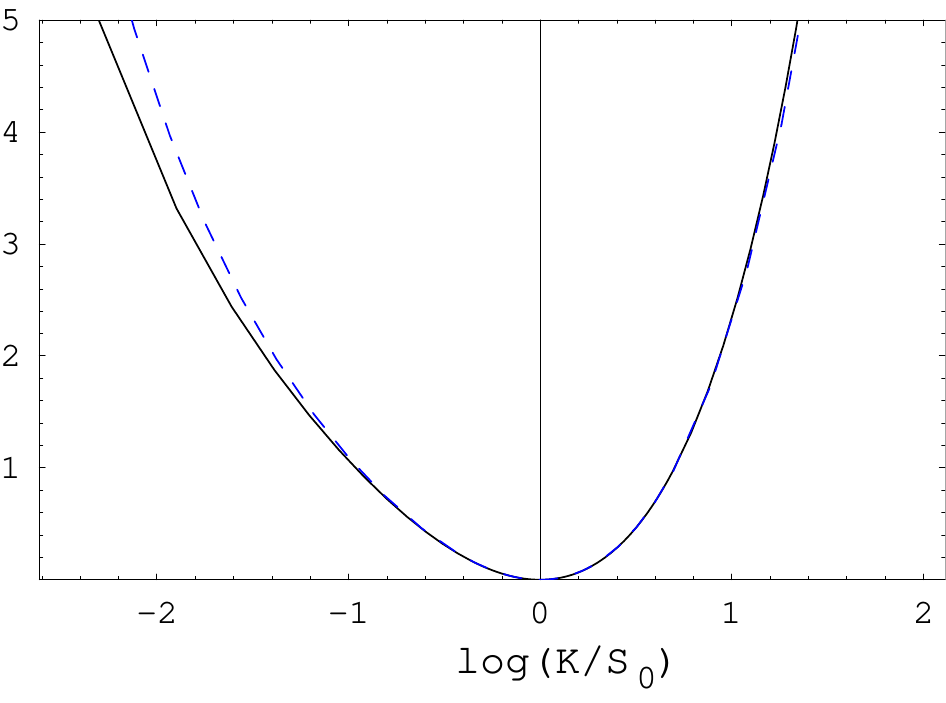}
    \caption{ The rate function $\mathcal{I}(K,S_0)$ for $\beta=\frac12$
in units of $S_0/\sigma^2$ vs $K$ (left) and vs $\log(K/S_0)$ (right)
(solid black curve) and the Taylor expansion (\ref{Taylor}) keeping the first 
three terms (dashed blue).
}
\label{Fig:RateFunction}
 \end{figure}
\section{Asian options in the CEV model}
\label{Sec:3}

The CEV model is defined by the one-dimensional diffusion under 
the risk-neutral measure
\begin{equation}
dS_{t}=(r-q)S_{t}dt+\sigma S_{t}^{\beta}dW_{t},
\end{equation}
with $S_{0}>0$. 

It is easy to check that the following Lemma holds. 

\begin{lemma}\label{lemma:1}
For an Asian OTM call option, that is, $K > S_{0}$, we have for 
$\frac12 \leq \beta < 1$
\begin{equation}
\lim_{T\rightarrow 0}T\log C(T)=
\lim_{T\rightarrow 0}T\log\mathbb{P}
\left(\frac{1}{T}\int_{0}^{T}S_{t}dt\geq K\right).
\end{equation}

For an Asian OTM put option, that is, $K \leq S_{0}$, we have for 
$\frac12 \leq \beta < 1$
\begin{equation}
\lim_{T\rightarrow 0}T\log P(T)=
\lim_{T\rightarrow 0}T\log\mathbb{P}
\left(\frac{1}{T}\int_{0}^{T}S_{t}dt\leq K\right).
\end{equation}
\end{lemma}

Using this result we can prove the short maturity asymptotics for 
OTM Asian options in the CEV model (\ref{CEVdef}). 

\begin{theorem}\label{ThmOTMCEV}
The short maturity asymptotics for OTM Asian options 
in the CEV model (\ref{CEVdef}) with $\frac12 \leq \beta < 1$
is given by
\begin{equation}
\lim_{T\rightarrow 0}T\log C(T)= - \mathcal{I}(K,S_0)\,,
\end{equation}
where the rate function is given by the solution of a variational problem
specified as follows.

(i) For OTM Asian call options $K >S_0$ we have
\begin{equation}
\mathcal{I}(K,S_0) = 
\inf_{\int_{0}^{1}g(t)dt\geq K, g(0)=S_{0}, g(t)\geq 0, 0\leq t\leq 1}
\frac{1}{2}\int_{0}^{1}\frac{(g'(t))^{2}}{\sigma^{2}g(t)^{2\beta}}dt\,,\quad
K > S_0\,.
\end{equation}

(ii) For OTM Asian put options $K < S_0$ we have 
\begin{equation}
\mathcal{I}(K,S_0) = 
\inf_{\int_{0}^{1}g(t)dt\leq K, g(0)=S_{0}, g(t)\geq 0, 0\leq t\leq 1}
\frac{1}{2}\int_{0}^{1}\frac{(g'(t))^{2}}{\sigma^{2}g(t)^{2\beta}}dt\,,\quad
K < S_0\,.
\end{equation}
\end{theorem}

\subsection{At-the-Money Asian Options}
\label{ATMSection}

Let us consider the ATM case, that is, $K=S_{0}>0$. 
For this case we have the following result.

\begin{theorem}\label{ThmATMSqrt}
As $T\rightarrow 0$, we have in the CEV model with $\frac12\leq \beta < 1$
\begin{equation}
C(T)=
\sigma S_{0}^{\beta}\frac{\sqrt{T}}{\sqrt{6\pi}}+O(T),
\qquad
P(T)=
\sigma S_{0}^{\beta}\frac{\sqrt{T}}{\sqrt{6\pi}}+O(T).
\end{equation}
\end{theorem}

\subsection{Variational Problem for Short-Maturity Asymptotics for Asian Options
in the CEV model}
\label{Sec:3.2}


Theorem~\ref{ThmOTMCEV} gives the rate function $\mathcal{I}(K,S_0)$
of an Asian option in the CEV model as a variational problem. 
For OTM Asian call option $ K > S_0$ this variational
problem reads
\begin{equation}\label{Ivarproblem}
\mathcal{I}(K,S_0) = \mbox{inf}_g \frac{1}{2\sigma^2}
\int_0^1  \frac{(g'(t))^2}{g(t)^{2\beta}} dt\,,
\end{equation}
where the function $g(t)$ is differentiable and satisfies $g(0)=S_0$, $g(t)>0, 0 \leq t \leq 1$ 
and the infimum is taken under the constraint
\begin{equation}\label{constineq}
\int_0^1 g(t) dt \geq K\,.
\end{equation}
Similarly, for OTM Asian put option with $K < S_0$, the rate function
$\mathcal{I}(K,S_0)$ is given by the variational problem (\ref{Ivarproblem})
with inequality constraint $\int_0^1 g(t) dt \leq K$.

Define $\mathcal{I}_K(K,S_0)$ as the solution of the variational 
problem (\ref{Ivarproblem}), obtained by replacing the inequality 
(\ref{constineq}) with equality. The strategy of the proof will be to 
show that $\mathcal{I}_K(K,S_0)$ is an increasing 
function for $K>S_0$ and thus the solution of the variational inequality is
given by $\mathcal{I}(K,S_0) = \mathcal{I}_K(K,S_0)$. 
For $K<S_0$ we will show that 
$\mathcal{I}_K(K,S_0)$ is a decreasing function for $K<S_0$, and thus the 
solution of the variational inequality is given by $\mathcal{I}(K,S_0) = 
\mathcal{I}_K(K,S_0)$. 

We give next the solution of the variational problem (\ref{Ivarproblem})
with the equality constraint $\int_0^1 g(t) dt = K$. This is given by the
following result.

\begin{proposition}\label{VarProp}
The solution of the variational problem (\ref{Ivarproblem}) with the 
equality constraint $\int_0^1 g(t) dt = K$ is given by
\begin{equation}\label{Iresult}
\mathcal{I}_K(K, S_0)=
\begin{cases}
\frac{S_0^{2(1-\beta)}}{2\sigma^2 }
a^{(+)}(x) b^{(+)}(x) & K \leq S_{0}\,, \\
\frac{S_0^{2(1-\beta)}}{2\sigma^2 }
a^{(-)}(x) b^{(-)}(x) & K \geq S_{0}\,.
\end{cases}
\end{equation}
The two cases are as follows:

(i) $K \leq S_0$. 
$0 < x \leq 1$ is the solution of the equation
\begin{equation}\label{ratecase1}
\frac{K}{S_0} = 
x + \frac{b^{(+)}(x)}{a^{(+)}(x)}\,,
\end{equation}
with
\begin{align}
a^{(+)}(x) &= 2x^{-\beta} (1-x)^{\frac12}\,
{}_2F_1\left( \beta, \frac12; \frac32; 1 - \frac{1}{x} \right)\,, \\
b^{(+)}(x) &= \frac23 x^{-\beta} (1-x)^{\frac32}\,
{}_2F_1\left( \beta, \frac32; \frac52; 1 - \frac{1}{x} \right)\,.
\end{align}
The argument $z=1-\frac{1}{x}$ of the hypergeometric function 
${}_2F_1(a,b;c;z)$ is negative.

(ii) $K \geq S_0$. 
$x \geq 1$ is the solution of the equation
\begin{equation}\label{ratecase2}
\frac{K}{S_0} = 
x - \frac{b^{(-)}(x)}{a^{(-)}(x)}\,,
\end{equation}
with
\begin{align}
a^{(-)}(x)
&=2 x^{-\beta} (x-1)^{\frac12}
{}_2F_1\left( \beta, \frac12; \frac32; 1-\frac{1}{x} \right)\,, \\
b^{(-)}(x)&=
\frac23 x^{-\beta} (x-1)^{\frac32}
{}_2F_1\left( \beta, \frac32; \frac52; 1-\frac{1}{x}  \right)\,.
\end{align}
The argument $z=1-\frac{1}{x}$ of the hypergeometric function 
${}_2F_1(a,b;c;z)$ is positive.

\end{proposition}

An alternative form of the solution for $\mathcal{I}_K(K,S_0)$
which gives additional information on the
continuity and monotonicity properties of this function in $K$
is given by the following result.

\begin{proposition}\label{prop:IKalt}

(i) The function $\mathcal{I}_K(K,S_0)$ is given for $K> S_0$ by 
\begin{equation}\label{IK1alt}
\mathcal{I}_K(K,S_0) = \inf_{\varphi > K/S_0} \frac12 
\frac{[\mathcal{G}^{(-)}(\varphi)]^2}{\varphi - \frac{K}{S_0}}\,,
\end{equation}
with 
\begin{eqnarray}
\mathcal{G}^{(-)}(\varphi) &=& \frac{S_0^{1-\beta}}{\sigma}
\int_1^\varphi z^{-\beta} \sqrt{\varphi-z} dz \\
&=& \frac{S_0^{1-\beta}}{\sigma} \frac23 \varphi^{-\beta}
(\varphi-1)^{3/2} 
{}_2F_1\left(\frac32, \beta;\frac52;1 - \frac{1}{\varphi}\right)\,.\nonumber
\end{eqnarray}

(ii) The function $\mathcal{I}_K(K,S_0)$ is given for $K< S_0$ by 
\begin{equation}\label{IK2alt}
\mathcal{I}_K(K,S_0) = \inf_{0 < \chi < K/S_0} \frac12 
\frac{[\mathcal{G}^{(+)}(\chi)]^2}{\frac{K}{S_0} - \chi}\,,
\end{equation}
with 
\begin{eqnarray}
\mathcal{G}^{(+)}(\chi) &=& \frac{S_0^{1-\beta}}{\sigma}
\int_\chi^1 z^{-\beta} \sqrt{z-\chi} dz \\
&=& \frac{S_0^{1-\beta}}{\sigma} \frac23 \chi^{-\beta}
(1-\chi)^{3/2} 
{}_2F_1\left(\frac32, \beta;\frac52;1 - \frac{1}{\chi}\right)\,.\nonumber
\end{eqnarray}
\end{proposition}

From the representation of Proposition \ref{prop:IKalt}
it follows that $\mathcal{I}_K(K,S_0)$ is a continuous
function of $K$. 
We also obtain the 
monotonicity properties of this function, which imply the relation
to the rate function $\mathcal{I}(K,S_0)$ given by Theorem~\ref{ThmOTMCEV}.

\begin{corollary}\label{cor:mono}
We have the following monotonicity properties of the function 
$\mathcal{I}_K(K,S_0)$ with respect to strike $K$:

(i) For $K> S_0$ the function $\mathcal{I}_K(K,S_0)$ is an increasing
function of $K$.

(ii) For $K< S_0$ the function $\mathcal{I}_K(K,S_0)$ is a decreasing
function of $K$.

(iii) The rate function $\mathcal{I}(K,S_0)$ is given by
\begin{equation}
\mathcal{I}(K,S_0) = \mathcal{I}_K(K,S_0)\,.
\end{equation}
\end{corollary}

\begin{remark}
For $\beta=\frac12$, the results 
of Proposition~\ref{VarProp} recover the result of Proposition~\ref{prop:2}.
For this case the hypergeometric
functions can be expressed in terms of elementary functions. 

(i) For $K<S_0$ we need the expressions of $a^{(+)}(x), b^{(+)}(x)$
for $x$ real and negative. We have for $z \in \mathbb{R}_+$
\begin{eqnarray}
{}_2 F_1\left(\frac12, \frac12;\frac32; -z\right) &=& 
\frac{\mbox{arcsinh} \sqrt{z}}{\sqrt{z}}\,, \\
{}_2 F_1\left(\frac12, \frac32;\frac52; -z\right) &=& 
-\frac32 \frac{\mbox{arcsinh} \sqrt{z}}{z^{3/2}} + \frac32 \frac{\sqrt{1+z}}{z}\,.
\end{eqnarray}
The equation (\ref{ratecase1}) reads
\begin{equation}
\frac{K}{S_0} = \frac12 x + \frac12 
\frac{\sqrt{1-x}}{\mbox{arcsinh}\frac{\sqrt{1-x}}{\sqrt{x}}}\,.
\end{equation}
Denoting $x \to  \frac{1}{\cosh^4 x}$ this equation becomes 
identical to (\ref{19}). 

The rate function (\ref{Iresult}) is
\begin{equation}
\mathcal{I}(K,S_0) = \frac{S_0}{\sigma^2} \left\{
- x \mbox{arcsinh}^2 \frac{\sqrt{1-x}}{\sqrt{x}} +
\sqrt{1- x} \mbox{arcsinh} \frac{\sqrt{1-x}}{\sqrt{x}}\right\}\,.
\end{equation}
Substituting here again $x \to \frac{1}{\cosh^4 x}$ this becomes identical 
with the result (\ref{18}) for the rate function for $K<S_0$ in the 
square-root model.

(ii) A similar argument holds for $K>S_0$ using the expressions for
the hypergeometric functions of positive argument
\begin{eqnarray}
{}_2 F_1\left(\frac12, \frac12;\frac32; z\right) &=& \frac{\arcsin \sqrt{z}}{\sqrt{z}}\,, \\
{}_2 F_1\left(\frac12, \frac32;\frac52; z\right) &=& 
\frac32 \frac{\arcsin \sqrt{z}}{z^{3/2}} - \frac32 \frac{\sqrt{1-z}}{z}\,.
\end{eqnarray}
\end{remark}


\begin{remark}
For $\beta\rightarrow 1$, the results 
of Proposition~\ref{VarProp} recover the rate function for the Black-Scholes
model in Proposition~12 of \cite{ShortMatAsian}.

(i) For $K<S_0$ we need the expressions of $a^{(+)}(x), b^{(+)}(x)$
for $x\in (0,1]$. We have for $z = \frac{1}{x}-1 \in \mathbb{R}_+$
\begin{eqnarray}
{}_2 F_1\left(1, \frac12;\frac32; -z\right) &=& 
\frac{\arctan \sqrt{z}}{\sqrt{z}}\,, \\
{}_2 F_1\left(1, \frac32;\frac52; -z\right) &=& 
\frac{3}{z}-3 \frac{\arctan \sqrt{z}}{z^{3/2}} \,.
\end{eqnarray}
The equation (\ref{ratecase1}) reads
\begin{equation}
\frac{K}{S_0} = \frac{\sqrt{x(1-x)}}
{\arctan\sqrt{1/x-1}} \,.
\end{equation}
Identifying $\arctan\sqrt{1/x-1}=\xi$ this reproduces Eq.~(33) in
\cite{ShortMatAsian}
\begin{equation}
\frac{K}{S_0} = \frac{1}{2\xi} \sin(2\xi)\,.
\end{equation}

The rate function (\ref{Iresult}) becomes, when expressed in terms of $\xi$
\begin{equation}
\mathcal{I}(K,S_0) = \frac{1}{\sigma^2} 
2\xi (\tan \xi - \xi)\,,
\end{equation}
which reproduces
the result Eq.~(31) in \cite{ShortMatAsian} for the rate function for $K<S_0$ in the BS model.

(ii) A similar argument holds for $K>S_0$ using the expressions for
the hypergeometric functions of positive argument, $x\geq 1$
\begin{eqnarray}
{}_2 F_1\left(1, \frac12;\frac32; z\right) &=& 
\frac{\mbox{arctanh} \sqrt{z}}{\sqrt{z}}\,, \quad z = 1 - \frac{1}{x}\,,\\
{}_2 F_1\left(1, \frac32;\frac52; z\right) &=& 
-\frac{3}{z^{3/2}}(\sqrt{z} - \mbox{arctanh} \sqrt{z}) \,.
\end{eqnarray}
Identifying $\hat{\beta} = \frac12 \mbox{arctanh}\sqrt{1-\frac{1}{x}}$ we get the
rate function
\begin{equation}
\mathcal{I}(K,S_0) = \frac{1}{\sigma^2}\left(\frac12\hat{\beta}^2 - \hat{\beta} \tanh
(\hat{\beta}/2)\right)\,,
\end{equation}
where $\hat{\beta}$ is the solution of the
equation $\frac{K}{S_0}=\frac{1}{2\hat{\beta}}\sinh(2\hat{\beta})$.
These are identical with the results of Proposition 12 of \cite{ShortMatAsian}.
\end{remark}
\subsection{Expansion of the rate function around the ATM point}

Using the same approach as in the proof of Proposition~14
in \cite{ShortMatAsian} one can expand the
rate function in power series of $x=\log(K/S_0)$ for arbitrary
$\beta$. The first few terms are
\begin{eqnarray}\label{Taylor}
\mathcal{I}(K,S_0) &=& \frac{S_0^{2(1-\beta)}}{\sigma^2}
\left\{ \frac32 x^2 +\left( -\frac{3}{10} + \frac95 (1-\beta)
\right) x^3 \right. \\
& & \left. + 
\left( \frac{109}{1400} - \frac{117}{350}(1-\beta) + \frac{198}{175}(1-\beta)^2
\right) x^4 + O(x^5)\right\}\,. \nonumber
\end{eqnarray}
For $\beta=\frac12$ this reduces to the expansion of the rate function
in the square-root model given in equation (\ref{SqrtTaylor}).

\subsection{Asymptotics of the rate function}

We discuss next the asymptotics of the rate function $\mathcal{I}(K,S_0)$ in 
the CEV model for very small/large strike $K$. This is given by
the following result, which generalizes the results of Proposition~\ref{prop:4}
to general $\frac12\leq \beta < 1$.

\begin{proposition}[Large strike asymptotics]\label{prop:11}
We have, for $\beta \in [\frac{1}{2},1)$,
\begin{equation}
\mathcal{I}(K,S_0)
\sim
\frac{S_0^{2(1-\beta)}}{2\sigma^2} 
\frac{\pi \Gamma^2(1-\beta)}{(3-2\beta) \Gamma^2(3/2-\beta)}
\left( \frac{3-2\beta}{2(1-\beta)} \frac{K}{S_0}\right)^{2(1-\beta)}\,,
\qquad\text{as $K\rightarrow\infty$},
\end{equation}
where $\Gamma(\cdot)$ is the Gamma function.

For $\beta=\frac12$ this reproduces the result (i) of Proposition~\ref{prop:4}.
\begin{equation}
\mathcal{I}(K,S_0)
\sim
\frac{\pi^2 K}{2\sigma^2} \,,
\qquad\text{as $K\rightarrow\infty$}.
\end{equation}
\end{proposition}

\begin{proposition}[Small strike asymptotics]\label{prop:12}
The $K\to 0$ asymptotics of the rate function for $\beta \in(\frac{1}{2},1)$
is given by
\begin{equation}
\lim_{K\to 0} \frac{K}{S_0} \mathcal{I}(K,S_0) = 
\frac{2S_0^{2(1-\beta)}}{\sigma^2(3-2\beta)^2} \,.
\end{equation}

For $\beta\rightarrow\frac12$, this reproduces the result (ii) of 
Proposition~\ref{prop:4}
\begin{equation}
\lim_{K\to 0} \frac{K}{S_0} \mathcal{I}(K,S_0)= 
\frac{2S_0^{2(1-\beta)}}{\sigma^2(3-2\beta)^2}
\rightarrow
\frac{S_0}{2\sigma^2}\,,
\end{equation}
\end{proposition}


\begin{figure}[b!]
    \centering
   \includegraphics[width=3.5in]{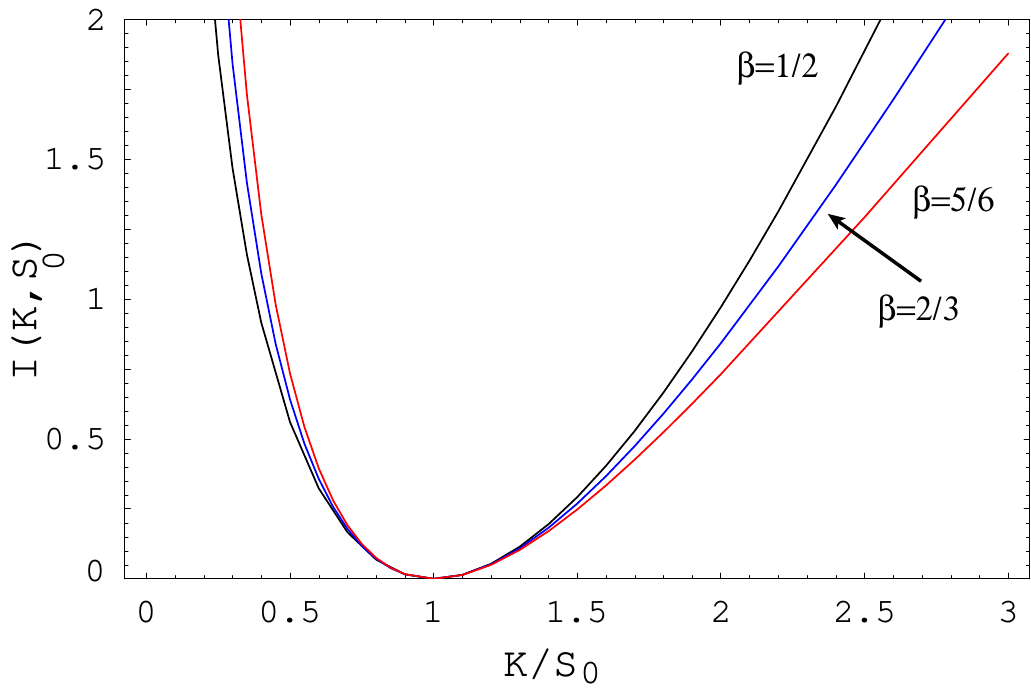}
    \caption{ 
The rate function $\mathcal{I}(K,S_0)/(S_0^{2(1-\beta)}/\sigma^2)$ vs $K/S_0$
for Asian options in the CEV model with $\beta=\frac12$ (black),
$\beta=\frac23$ (blue) and $\beta=\frac56$ (red).
}
\label{Fig:IK}
 \end{figure}
 

\section{Floating Strike Asian Options}
\label{Sec:4}

We consider in this Section the short maturity asymptotics for 
floating strike Asian options.
The prices of the floating strike Asian call/put options are given by
risk-neutral expectations
\begin{align}
&C_f(T):=e^{-rT}\mathbb{E}\left[\left(\kappa S_{T}-\frac{1}{T}\int_{0}^{T}S_{t}dt\right)^{+}\right],
\\
&P_f(T):=e^{-rT}\mathbb{E}\left[\left(\frac{1}{T}\int_{0}^{T}S_{t}dt-\kappa S_{T}\right)^{+}\right].
\end{align}
First of all, similar to Lemma~\ref{lemma:1}, we have:

(i) For an Asian OTM call option, that is, $\kappa<1$, we have for 
$\frac12 \leq \beta < 1$
\begin{equation}\label{CfEquiv}
\lim_{T\rightarrow 0}T\log C_{f}(T)=
\lim_{T\rightarrow 0}T\log\mathbb{P}
\left(\frac{1}{T}\int_{0}^{T}S_{t}dt\leq \kappa S_{T}\right).
\end{equation}

(ii) For an Asian OTM put option, that is, $\kappa>1$, we have for 
$\frac12 \leq \beta < 1$
\begin{equation}\label{PfEquiv}
\lim_{T\rightarrow 0}T\log P_{f}(T)=
\lim_{T\rightarrow 0}T\log\mathbb{P}
\left(\frac{1}{T}\int_{0}^{T}S_{t}dt\geq \kappa S_{T}\right).
\end{equation}

We start by considering the square-root model:
\begin{equation}
dS_{t}=(r-q)S_{t}dt+\sigma\sqrt{S_{t}}dW_{t},
\end{equation}
with $S_{0}>0$ and $W_{t}$ is a standard Brownian motion starting at zero 
at time zero.

\begin{theorem}\label{ThmFloatingStrikeSqrt2}
For $\beta=\frac{1}{2}$, $\mathbb{P}\left(\frac{1}{T}\int_{0}^{T}S_{t}dt-\kappa S_{T}\in\cdot\right)$
satisfies a large deviation principle with the rate function 
\begin{equation}\label{IfSqrt}
I_{f}(x):=
\sup_{\theta\in\mathbb{R}}\{\theta x-\Lambda_{f}(\theta)\}
\end{equation}
where $\Lambda_{f}(\theta)$ is given by
\begin{align}
\Lambda_{f}(\theta)&:=
\lim_{T\rightarrow 0}T\log\mathbb{E}\left[e^{\frac{\theta}{T^{2}}\int_{0}^{T}S_{t}dt}\right]
\\
&=
\begin{cases}
\frac{\sqrt{2\theta}}{\sigma}
\tan\left(\frac{\sigma}{2}\sqrt{2\theta}+\tan^{-1}\left(-\sigma\kappa\sqrt{\frac{\theta}{2}}\right)\right)S_{0} &\text{if $0\leq\theta<\theta_{c}$}
\\
-\frac{\sqrt{-2\theta}}{\sigma}
\tanh\left(\frac{\sigma}{2}\sqrt{-2\theta}+\tanh^{-1}\left(-\sigma\kappa\sqrt{\frac{-\theta}{2}}\right)\right)S_{0} 
&\text{if $\theta\leq 0$}
\\
+\infty &\text{otherwise}
\end{cases}.
\nonumber
\end{align}
where $\theta_c$ is the unique positive solution of the equation
\begin{equation}
\sqrt{\frac{\sigma^{2}\theta_{c}}{2}}+
\tan^{-1}\left(-\sigma\kappa\sqrt{\frac{\theta_{c}}{2}}\right)=\frac{\pi}{2}\,.
\end{equation}
It follows from \eqref{CfEquiv} and \eqref{PfEquiv}
that for $\kappa<1$, the call option is OTM and $C_{f}(T)=e^{-\frac{1}{T}\mathcal{I}_f(\kappa, S_0) +o(1/T)}$,
as $T\rightarrow 0$,
and for $\kappa>1$, the put option is OTM and $P_{f}(T)=e^{-\frac{1}{T}\mathcal{I}_f(\kappa, S_0) +o(1/T)}$,
as $T\rightarrow 0$,
where
\begin{equation}
\mathcal{I}_f(\kappa, S_0)=I_{f}(0)=\sup_{\theta\in\mathbb{R}}\{-\Lambda_{f}(\theta)\}.
\end{equation}
\end{theorem}

The result of Theorem \ref{ThmFloatingStrikeSqrt2} for $\mathcal{I}_f(\kappa,S_0)$
for the square root model can be put into a more explicit form, as
\begin{eqnarray}\label{Jfdef}
\mathcal{I}_f(\kappa, S_0) = \frac{S_0}{\sigma^2} \mathcal{J}_f(\kappa)\,,
\end{eqnarray}
where $\mathcal{J}_f(\kappa)$ is given by:

(i) For $\kappa\geq 1$ 
\begin{eqnarray}\label{Jfplus}
\mathcal{J}_f(\kappa) = 2z \frac{\kappa z-\tan z}{1 + \kappa z \tan z}\,,
\end{eqnarray}
where $z$ is the solution of the equation
\begin{eqnarray}\label{eqz1}
1+\kappa^2 z^2 + (1-\kappa^2 z^2) \frac{\sin 2z}{2z} = 2\kappa \cos^2 z\,.
\end{eqnarray}
The solution is defined up to a sign, but this ambiguity is not relevant
for computing $\mathcal{J}_f(\kappa)$.

(ii) For $\kappa\leq 1$ 
\begin{eqnarray}\label{Jf2}
\mathcal{J}_f(\kappa) = 2z \frac{\kappa z-\tanh z}
{1 - \kappa z \tanh z}\,,
\end{eqnarray}
where $z$ is the solution of the equation
\begin{eqnarray}\label{eqz2}
1-\kappa^2 z^2 + (1+\kappa^2 z^2) \frac{\sinh 2z}{2z} = 2\kappa \cosh^2 z\,.
\end{eqnarray}
The rate function $\mathcal{J}_f(\kappa,S_0)$ for the square root model 
$\beta=\frac12$ is shown in Fig.~\ref{Fig:If} (solid black curve). 
This is compared  against the rate  function $\mathcal{I}(\kappa)$ for 
fixed strike Asian options given by Proposition~\ref{prop:2} (dashed curve). 
In the
Black-Scholes model they are equal \cite{ShortMatAsian}, which follows from the
equivalence relations for fixed/floating strike Asian options \cite{HeWo}.
These relations do not hold beyond the Black-Scholes model, and as a 
consequence the corresponding rate functions are different.

The floating strike rate function has the expansion around the ATM point
$\kappa=1$
\begin{eqnarray}\label{Jfexp}
\mathcal{J}_f(\kappa) = \frac32 \log^2\kappa - \frac{33}{20} \log^3 \kappa
+ \frac{5809}{5600} \log^4\kappa + O(\log^5\kappa) \,.
\end{eqnarray}
This is obtained by expanding the solution of the equations (\ref{eqz1}), 
(\ref{eqz2}) in powers of $z$, and inserting the result into (\ref{Jfplus}) 
and (\ref{Jf2}).
The approximation for the rate function $\mathcal{J}_f(\kappa)$ obtained
by keeping the first three terms in this expansion is shown in 
Figure~\ref{Fig:If} as the solid blue curve. 


\begin{figure}[b!]
    \centering
   \includegraphics[width=3.5in]{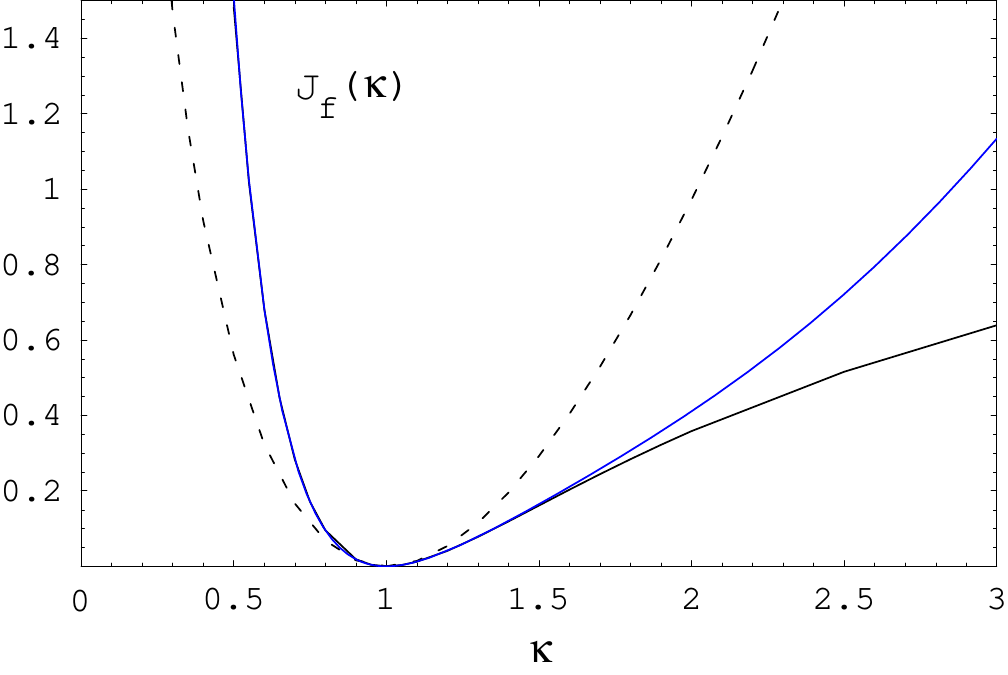}
    \caption{ 
The rate function $\mathcal{J}_f(\kappa)$ vs $\kappa$
for floating strike Asian options in the square root model with 
$\beta=\frac12$ (black solid curve). The solid blue curve shows the
approximation of this function obtained by keeping the first 3 terms in the
expansion (\ref{Jfexp}). This is compared against the fixed strike
rate function $\mathcal{I}(S_0 \kappa, S_0)$ (in units of $S_0/\sigma^2$) 
for the same model (dashed curve), given by Proposition~\ref{prop:2}.
}
\label{Fig:If}
 \end{figure}
 

For the general CEV model with $\frac{1}{2}\leq\beta<1$, 
following the proof of Theorem~\ref{ThmOTMCEV}, we get the following result:

\begin{theorem}\label{ThmFloatingStrikeSqrt}
The short maturity asymptotics for OTM floating strike Asian options 
in the CEV model (\ref{CEVdef}) with $\frac12 \leq \beta < 1$ is given by

(i) For $\kappa<1$, the short maturity asymptotics for OTM
floating strike Asian call option is
\begin{equation}
\lim_{T\rightarrow 0}T\log C_{f}(T)
=-\mathcal{I}_f(\kappa, S_0),
\end{equation}
where
\begin{equation}
\mathcal{I}_f(\kappa, S_0)=\inf_{\int_{0}^{1}g(t)dt\leq\kappa g(1), g(0)=S_{0}, g(t)\geq 0, 0\leq t\leq 1}
\frac{1}{2}\int_{0}^{1}\frac{(g'(t))^{2}}{\sigma^{2}g(t)^{2\beta}}dt.
\end{equation}

(ii) For $\kappa>1$, the short maturity asymptotics for OTM
floating strike Asian put option is
\begin{equation}
\lim_{T\rightarrow 0}T\log P_{f}(T)
=-\mathcal{I}_f(\kappa, S_0),
\end{equation}
where
\begin{equation}
\mathcal{I}_f(\kappa, S_0)=\inf_{\int_{0}^{1}g(t)dt\geq\kappa g(1), g(0)=S_{0}, g(t)\geq 0, 0\leq t\leq 1}
\frac{1}{2}\int_{0}^{1}\frac{(g'(t))^{2}}{\sigma^{2}g(t)^{2\beta}}dt.
\end{equation}
\end{theorem}

Let us consider the ATM case, that is, $\kappa=1$. 
For this case we have the following result.
The proof is very similar to the proof of Theorem~\ref{ThmATMSqrt}, and is hence omitted here.

\begin{theorem}
As $T\rightarrow 0$, we have in the CEV model with $\frac12\leq \beta < 1$,
\begin{equation}
C_{f}(T)=
\sigma S_{0}^{\beta}\frac{\sqrt{T}}{\sqrt{6\pi}}+O(T),
\qquad
P_{f}(T)=
\sigma S_{0}^{\beta}\frac{\sqrt{T}}{\sqrt{6\pi}}+O(T).
\end{equation}
\end{theorem}


\section{Numerical Tests}
\label{Sec:5}

We present in this Section a few numerical tests of the short-maturity
asymptotic results for Asian options in the CEV model obtained in this paper. 
Following \cite{ShortMatAsian} we will use the Asian option pricing formulas
\begin{eqnarray}\label{Capp}
C_{\rm asympt}(K,T) &=& e^{-rT} (A(T) N(d_1) - K N(d_2))\,,  \\
\label{Papp}
P_{\rm asympt}(K,T) &=& e^{-rT} (K N(-d_2) - A(T) N(- d_1))\,, 
\end{eqnarray}
where $A(T)$ is the expectation of the averaged asset price, 
\begin{equation}
A(T) = S_0 \frac{1}{(r-q)T}(e^{(r-q)T}-1) \,,
\end{equation}
and
\begin{equation}
d_{1,2} = \frac{1}{\Sigma_{LN}\sqrt{T}}
\left( \log\frac{A(T)}{K} \pm \frac12 \Sigma^2_{LN} T \right)\,.
\end{equation}
The equivalent log-normal volatility of the Asian option is defined by 
\begin{equation}\label{SigLN}
\Sigma_{LN}^2(K,S_0) = \frac{\log^2(K/S_0)}{2\mathcal{I}(K,S_0)} \,,
\end{equation}
where $\mathcal{I}(K,S_0)$ is the rate function, given for the general
CEV model in (\ref{Iresult}), and for the square-root model $\beta=\frac12$
in Proposition~\ref{prop:2}. As shown in Proposition 18 of 
\cite{ShortMatAsian}, the approximation (\ref{Capp}),(\ref{Papp}) 
has the same short maturity asymptotics as that given by
Proposition~\ref{prop:2} for the square-root model $\beta=\frac12$, and by 
Theorem~\ref{ThmOTMCEV} for the general CEV model. 

\subsection{Equivalent log-normal volatility of Asian options in the CEV model}

The series expansion of the equivalent log-normal volatility 
$\Sigma_{\rm LN}(K,S_0)$ in powers of log-strike $x = \log(K/S_0)$ can be 
obtained by substituting (\ref{Taylor}) into the definition (\ref{SigLN}). 
This is
\begin{eqnarray}\label{SigLNCEVexp}
\Sigma_{LN}(K,S_0) &=& \sigma \frac{1}{\sqrt{3}} S_0^{\beta-1}
\left\{ 1 + \left( \frac{1}{10} + \frac35 (\beta-1) \right) x \right. \\
& & \left. + 
\left( - \frac{23}{2100} + \frac{12}{175} (\beta-1) + \frac{57}{350} (\beta-1)^2
\right) x^2 + O(x^3)
\right\} \,.\nonumber
\end{eqnarray}
For ATM Asian options $K=S_0$ the equivalent log-normal volatility is
\begin{equation}\label{SigATM}
\Sigma_{LN}(S_0,S_0) = \sigma \frac{1}{\sqrt{3}} S_0^{\beta-1}\,.
\end{equation}

For the square-root model $\beta=\frac12$ the ATM skew and convexity of the 
equivalent log-normal volatility are $-\frac{1}{5}$ and $-\frac{19}{4200}$ of
the ATM equivalent volatility, respectively. 
We show in Figure~\ref{Fig:SigmaBS} the plot of $\Sigma_{\rm LN}/\sigma$
vs $x=\log(K/S_0)$ obtained using equation (\ref{SigLN}) for the square-root 
model $\beta=\frac12$.

The plot in Figure~\ref{Fig:SigmaBS} is in general qualitative agreement with 
the shape of the implied volatility 
for options on realized variance in the Heston model, which are mathematically
equivalent to Asian options in the square-root model. See Figure 5 (right) in 
\cite{Drimus}. 
As noted in the literature \cite{Drimus}, the down-sloping shape of the
implied volatility is a deficiency of the Heston model, as 
the observed smile for variance options in equity markets is up-sloping. 
The reference \cite{Drimus} proposes as an alternative model which has 
up-sloping smiles for variance options, the 3/2 stochastic volatility model
\cite{CarrSun,CarrWu}.

From the expansion (\ref{SigLNCEVexp}) one can obtain the dependence of the
ATM skew and convexity on the $\beta$ parameter. For $\beta=\frac12$ the ATM skew
is negative; as $\beta$ is increased, the ATM skew increases, crosses
zero at $\beta=\frac56$ and becomes positive. The ATM convexity is always 
negative for $\frac12 \leq \beta < 1$, so the equivalent log-normal volatility
smile is slightly concave. 

\begin{figure}[b!]
\centering
\includegraphics[width=3.5in]{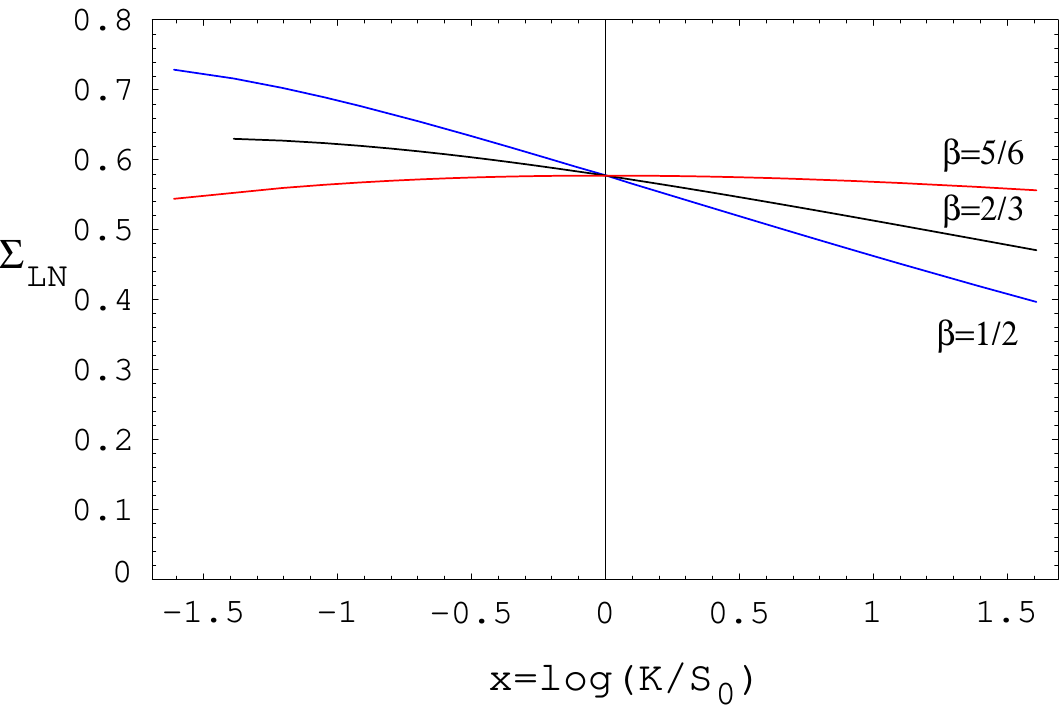}
\caption{Small-maturity equivalent log-normal volatility 
$\Sigma_{\rm LN}(K,S_0)/(\sigma S_0^{\beta-1})$ vs $x=\log(K/S_0)$
for an Asian option in the CEV model with $\beta=\frac12$ (blue),
$\beta=\frac23$ (black) and $\beta=\frac56$ (red).}
\label{Fig:SigmaBS}
 \end{figure}

\subsection{Numerical scenarios}

We present next numerical tests for Asian option pricing in the 
square-root model $\beta=\frac12$, for the 7 scenarios proposed in 
Dassios and Nagardjasarma \cite{DN2006}. 
We also compare with the third order approximation of
Foschi, Pagliarani, Pascucci \cite{FPP2013} (denoted as FPP3), listed in 
Table 5 of \cite{FPP2013}. The results are shown in Table~\ref{table1}. 

\begin{table}
\caption{Comparison of the short-maturity asymptotic formulas for Asian
options in the square-root model $\beta=\frac12$ for the 7 scenarios 
considered by Dassios and Nagardjasarma \cite{DN2006}. The results are
compared against those of \cite{DN2006} (DN) and those of Foschi, 
Pagliarani, Pascucci \cite{FPP2013} (denoted as FPP3). }
\begin{center}
\begin{tabular}{c|ccccc|c|cc}
\hline
Case & $S_0$ & $K$ & $r$ & $\sigma$ & $T$ & $C_{\rm asympt}(K,T)$ & DN & FPP3 \\
\hline
\hline
1 & 2 & 2 & 0.01 & 0.14 & 1 & 0.055474 & 0.0197 & 0.055562 \\
2 & 2 & 2 & 0.18 & 0.42 & 1 & 0.216013 & 0.2189 & 0.217874 \\
3 & 2 & 2 & 0.0125 & 0.35 & 2 & 0.170568 & 0.1725 & 0.170926 \\
4 & 1.9 & 2 & 0.05 & 0.69 & 1 & 0.189863 & 0.1902 & 0.190834 \\
5 & 2 & 2 & 0.05 & 0.72 & 1 & 0.250113 & NA & 0.251121 \\
6 & 2.1 & 2 & 0.05 & 0.72 & 1 & 0.307731 & 0.3098 & 0.308715 \\
7 & 2 & 2 & 0.05 & 0.71 & 2 & 0.350516 & 0.3339 & 0.353197 \\
\hline
\end{tabular}
\label{table1}
\end{center}
\end{table}

We note that the agreement of the asymptotic result with FPP3 is always better 
than 1\% in relative value.

The agreement improves dramatically if we define the equivalent log-normal 
volatility as
\begin{equation}
\Sigma_{LN}^2(K,S_0) = \frac{\log^2(K/A(T))}{2\mathcal{I}(K,A(T))} \,.
\end{equation}
However, the overall factor in $\mathcal{I}(K,S_0)$ must be still $S_0$, not 
$A(T)$, which is somewhat arbitrary. Therefore we do not use this approximation.
With this choice the agreement with FPP3 improves to better than 0.1\% in 
relative value.

A second set of scenarios proposed by DN \cite{DN2006} is shown in 
Table~\ref{table2}. 
There are 9 scenarios with $S_0=K=2,r=0.05,q=0,\beta=\frac12$.
The asymptotic results are shown in Table~\ref{table2}, comparing with
the results of \cite{DN2006} and \cite{FPP2013} (Table 6 in this reference).
Since they are all ATM scenarios, the use of the asymptotic formulas is 
very simple, and reduces to the use of equation (\ref{SigATM}).

The agreement of the asymptotic result with FPP3 is again very good, except 
for the $T=5Y$ case.
In all these cases (except $T=5Y$) the difference between them is less than 
1\% in relative value.
For maturities less than 1Y, the difference is always below 0.5\% in 
relative value.

\begin{table}
\caption{Numerical tests for the scenarios proposed in \cite{DN2006} and
\cite{FPP2013}.}
\begin{center}
\begin{tabular}{c|cc|c|cc}
\hline
Case & $\sigma$ & $T$ & $C_{\rm asympt}(K,T)$ & DN & FPP3 \\
\hline
\hline
1 & 0.71 & 0.1 & 0.075354 & 0.0751 & 0.075387 \\
2 & 0.71 & 0.5 & 0.172813 & 0.1725 & 0.173175 \\
3 & 0.71 & 1.0 & 0.247020 & 0.2468 & 0.248016 \\
4 & 0.71 & 2.0 & 0.350516 & 0.3339 & 0.353197 \\
5 & 0.71 & 5.0 & 0.536611 & 0.3733 & 0.545714 \\
\hline
6 & 0.1 & 1.0 & 0.061310 & 0.0484 & 0.061439 \\
7 & 0.3 & 1.0 & 0.120226 & 0.1207 & 0.120680 \\
8 & 0.5 & 1.0 & 0.181983 & 0.1827 & 0.182723 \\
9 & 0.7 & 1.0 & 0.243926 & 0.2446 & 0.244913 \\
\hline
\end{tabular}
\label{table2}
\end{center}
\end{table}

\subsection{Floating-strike Asian options}
We discuss also the pricing of floating-strike Asian options. 
They can be considered as call and put options on the underlying
$B_T := \kappa S_T - A_T$. The forward price of this asset is 
\begin{eqnarray}
F_f(T) := \mathbb{E}[B_T] = S_0 \left( \kappa e^{(r-q)T} -
\frac{e^{(r-q)T}-1}{(r-q)T} \right)\,.
\end{eqnarray}
For $\kappa \geq 0$, the underlying $B_T$ takes values on the entire
real axis. For this reason a Black-Scholes representation of this
asset is not appropriate.

We propose to approximate the prices of floating-strike Asian options
using a Bachelier (normal) approximation. These options are
approximated as zero strike put and call options on the asset
$B_T$, and their prices are
\begin{eqnarray}\label{CfBachelier}
C_f(\kappa,T) &=& e^{-rT} \left[
F_f(T) \Phi(d) + \frac{1}{\sqrt{2\pi}} \Sigma_N\sqrt{T} e^{-\frac12 d^2}
\right]\,, \\
P_f(\kappa,T) &=& e^{-rT} \left[
-F_f(T) \Phi(-d) + \frac{1}{\sqrt{2\pi}} \Sigma_N\sqrt{T} e^{-\frac12 d^2}
\right] \,, \nonumber
\end{eqnarray}
with $d=\frac{F_f(T)}{\Sigma_N\sqrt{T}}$. 

The equivalent normal volatility $\Sigma_N(\kappa,T)$ is specified by requiring 
that the small-maturity asymptotics of the floating-strike Asian options 
matches that of the Bachelier expression. This is given by the following 
result. 

\begin{proposition}
The short-maturity limit of the equivalent normal volatility in 
the square-root model $\beta=\frac12$ is given by:

(i) for OTM floating strike Asian options $\kappa \neq 1$
\begin{eqnarray}
\lim_{T\to 0} \Sigma_N(\kappa,T) = \frac{\sigma^2}{2S_0} \frac{(\kappa-1)^2}
{\mathcal J_f(\kappa)}\,,
\end{eqnarray}
where $\mathcal J_f(\kappa)$ is given by (\ref{Jfdef}).

(ii) for ATM floating strike Asian options $\kappa=1$
\begin{eqnarray}
\lim_{T\to 0} \Sigma_N(\kappa,T) = \sigma\sqrt{\frac{S_0}{3}}\,.
\end{eqnarray}
\end{proposition}

\begin{proof}
The proof is similar to that of Proposition~18 in \cite{ShortMatAsian} and 
will be  omitted.
\end{proof}

The pricing of floating-strike Asian options in the square-root model
has been considered in \cite{FMR}. This paper studied the pricing of
options with payoff $(-S_T + A_T - K)^+$ with $K$ both positive and
negative, using both discrete and continuous time monitoring. We will 
compare the result for $K=0$ with continuous time averaging, which 
corresponds in our notations to a floating strike Asian put option with 
$\kappa=1$.

The model parameters used in \cite{FMR} are $S_0=1,r=0.04, \sigma=0.7$, 
and the option maturity is $T=1$. The price quoted in Table 3 of this paper
with $K=0$ is $C_f(1,T) = 0.14376$. The asymptotic formula
(\ref{CfBachelier}) gives $C_f(1,T)=0.14524$, which is in reasonably good
agreement with the result of \cite{FMR} (1\% relative difference).

\section{Appendix: Proofs}
\label{sec:appendix}

\subsection{Background of Large Deviations Theory}

We start by giving a formal definition of the large deviation principle. 
We refer to Dembo and Zeitouni \cite{Dembo} 
for general background of large deviations theory and its applications. 

\begin{definition}[Large Deviation Principle]
A sequence $(P_{\epsilon})_{\epsilon\in\mathbb{R}^{+}}$ of probability measures on a topological space $X$ 
satisfies the large deviation principle with rate function $I:X\rightarrow\mathbb{R}$ if $I$ is non-negative, 
lower semicontinuous and for any measurable set $A$, we have
\begin{equation}
-\inf_{x\in A^{o}}I(x)\leq\liminf_{\epsilon\rightarrow 0}\epsilon\log P_{\epsilon}(A)
\leq\limsup_{\epsilon\rightarrow 0}\epsilon\log P_{\epsilon}(A)\leq-\inf_{x\in\overline{A}}I(x).
\end{equation}
Here, $A^{o}$ is the interior of $A$ and $\overline{A}$ is its closure. 
\end{definition}

The contraction principle plays a key role in our proofs. For the convenience
of the readers, we state the result as follows:

\begin{theorem}[Contraction Principle, e.g. Theorem 4.2.1. \cite{Dembo}]\label{Contraction}
If $P_{\epsilon}$ satisfies a large deviation principle on $X$ with rate 
function $I(x)$ and $F:X\rightarrow Y$ is a continuous map,
then the probability measures $Q_{\epsilon}:=P_{\epsilon}F^{-1}$ satisfies
a large deviation principle on $Y$ with rate function
\begin{equation}
J(y)=\inf_{x: F(x)=y}I(x).
\end{equation}
\end{theorem}

We will use the following version of the G\"{a}rtner-Ellis Theorem in the proofs in this paper.

\begin{theorem}[G\"{a}rtner-Ellis Theorem, e.g. Theorem \cite{Dembo}]\label{GEThm}
Let $Z_{\epsilon}$ be a sequence of random variables on $\mathbb{R}$. 
Assume the limit $\Lambda(\theta):=\lim_{\epsilon}\epsilon\log\mathbb{E}[e^{\frac{\theta}{\epsilon}Z_{\epsilon}}]$
exists on the extended real line and the interior of the set $\mathcal{D}:=\{\theta:\Lambda(\theta)<\infty\}$
contains $0$, and $\Lambda(\theta)$ is differentiable for any $\theta$ in the interior of $\mathcal{D}$
and $|\Lambda'(\theta)|\rightarrow\infty$ as $\theta$ approaches to the boundary of $\mathcal{D}$.
Then $\mathbb{P}(Z_{\epsilon}\in\cdot)$ satisfies a large deviation principle 
with the rate function $I(x):=\sup_{\theta\in\mathbb{R}}\{\theta x-\Lambda(\theta)\}$.
\end{theorem}

\subsection{Proofs of the results in Section~\ref{Sec:2}}


\begin{proof}[Proof of Theorem~\ref{ThmOTCSqrt}]
For any $\theta\in\mathbb{R}$, $u(t,x)=\mathbb{E}[e^{\theta\int_{0}^{t}S_{s}ds}|S_{0}=x]$
satisfies the PDE:
\begin{equation}
\frac{\partial u}{\partial t}=(r-q)x\frac{\partial u}{\partial x}
+\frac{1}{2}\sigma^{2}x\frac{\partial^{2}u}{\partial x^{2}}
+\theta xu(t,x),
\end{equation}
with $u(0,x)\equiv 1$. This affine PDE has the solution $u(t,x)=e^{A(t)x+B(t)}$, where
\begin{align}
&A'(t)=(r-q)A(t)+\frac{1}{2}\sigma^{2}A(t)^{2}+\theta,
\\
&B'(t)=0,
\end{align}
with $A(0)=B(0)=0$ and hence $B(t)=0$ and for $\theta>0$ sufficiently large,
\begin{equation}
\frac{2}{\sqrt{2\sigma^{2}\theta-(r-q)^{2}}}
\tan^{-1}\left(\frac{r-q+\sigma^{2}A}{\sqrt{2\sigma^{2}\theta-(r-q)^{2}}}\right)\bigg|^{A=A(t)}_{A=0}
=t,
\end{equation}
and thus
\begin{equation}
A(t;\theta)=\frac{\sqrt{2\sigma^{2}\theta-(r-q)^{2}}}{\sigma^{2}}
\tan\left[
\frac{\sqrt{2\sigma^{2}\theta-(r-q)^{2}}}{2}t
+\tan^{-1}\left(\frac{r-q}{\sqrt{2\sigma^{2}\theta-(r-q)^{2}}}\right)\right]
-\frac{r-q}{\sigma^{2}}.
\end{equation}

For $\theta<0$ sufficiently negative, 
\begin{equation}
\frac{2}{\sigma^{2}}
\frac{1}{2\sqrt{\frac{(r-q)^{2}}{\sigma^{4}}-\frac{2\theta}{\sigma^{2}}}}
\log\left(\frac{\frac{r-q}{\sigma^{2}}-\sqrt{\frac{(r-q)^{2}}{\sigma^{4}}-\frac{2\theta}{\sigma^{2}}}+A}{
\frac{r-q}{\sigma^{2}}+\sqrt{\frac{(r-q)^{2}}{\sigma^{4}}-\frac{2\theta}{\sigma^{2}}}+A}\right)
\bigg|^{A=A(t)}_{A=0}=t,
\end{equation}
and thus
\begin{align}
A(t;\theta)&=\frac{e^{t\sqrt{(r-q)^{2}-2\theta\sigma^{2}}}-1}
{\frac{1}{\frac{r-q}{\sigma^{2}}-\sqrt{\frac{(r-q)^{2}}{\sigma^{4}}-\frac{2\theta}{\sigma^{2}}}}
-\frac{e^{t\sqrt{(r-q)^{2}-2\theta\sigma^{2}}}}{\frac{r-q}{\sigma^{2}}+\sqrt{\frac{(r-q)^{2}}{\sigma^{4}}-\frac{2\theta}{\sigma^{2}}}}}
\\
&=\frac{\frac{2\theta}{\sigma^{2}}(e^{t\sqrt{(r-q)^{2}-2\theta\sigma^{2}}}-1)}{\frac{r-q}{\sigma^{2}}(1-e^{t\sqrt{(r-q)^{2}-2\theta\sigma^{2}}})+\sqrt{\frac{(r-q)^{2}}{\sigma^{4}}-\frac{2\theta}{\sigma^{2}}}(e^{t\sqrt{(r-q)^{2}-2\theta\sigma^{2}}}+1)}.
\nonumber
\end{align}

Let us study now the $T\to 0$ limit. We note that for any $T>0$ sufficiently 
small, we have
\begin{equation}
\mathbb{E}\left[e^{\frac{\theta}{T^{2}}\int_{0}^{T}S_{t}dt}\right]
=e^{A(T;\frac{\theta}{T^{2}})S_{0}}.
\end{equation}
For $0\leq\theta<\frac{\pi^{2}}{2\sigma^{2}}$,
\begin{equation}
\lim_{T\rightarrow 0}TA\left(T;\frac{\theta}{T^{2}}\right)=\sqrt{\frac{2\theta}{\sigma^{2}}}
\tan\sqrt{\frac{\sigma^{2}\theta}{2}},
\end{equation}
and this limit is $\infty$ if $\theta\geq\frac{\pi^{2}}{2\sigma^{2}}$. 

For $\theta<0$,
\begin{equation}
\lim_{T\rightarrow 0}TA\left(T;\frac{\theta}{T^{2}}\right)
=\frac{-\sqrt{-2\theta}}{\sigma}\frac{e^{\sigma\sqrt{-2\theta}}-1}{e^{\sigma\sqrt{-2\theta}}+1}
=\frac{-\sqrt{-2\theta}}{\sigma}\tanh\left(\frac{\sigma}{2}\sqrt{-2\theta}\right).
\end{equation}
Therefore,
\begin{equation}
\Lambda(\theta):=\lim_{T\rightarrow 0}T\log\mathbb{E}\left[e^{\frac{\theta}{T^{2}}\int_{0}^{T}S_{t}dt}\right]
=
\begin{cases}
\frac{\sqrt{2\theta}}{\sigma}
\tan\left(\frac{\sigma}{2}\sqrt{2\theta}\right)S_{0} &\text{if $0\leq\theta<\frac{\pi^{2}}{2\sigma^{2}}$}
\\
\frac{-\sqrt{-2\theta}}{\sigma}\tanh\left(\frac{\sigma}{2}\sqrt{-2\theta}\right)S_{0} &\text{if $\theta\leq 0$}
\\
+\infty &\text{otherwise}
\end{cases}.
\end{equation}
For $0<\theta<\frac{\pi^{2}}{2\sigma^{2}}$ and $\theta<0$, $\Lambda(\theta)$ is differentiable
and it is also easy to check that $\Lambda(\theta)$ is differentiable at $\theta=0$.
Finally, for $0<\theta<\frac{\pi^{2}}{2\sigma^{2}}$, we can compute that
\begin{equation}
\frac{\partial\Lambda(\theta)}{\partial\theta}
=\frac{\sqrt{2}}{\sigma 2\sqrt{\theta}}\tan\left(\frac{\sigma}{2}\sqrt{2\theta}\right)S_{0}
+\frac{\sqrt{2\theta}}{\sigma}\frac{\sigma\sqrt{2}}{4\sqrt{\theta}}\sec^{2}\left(\frac{\sigma}{2}\sqrt{2\theta}\right)S_{0}
\rightarrow+\infty,
\end{equation}
as $\theta\uparrow\frac{\pi^{2}}{2\sigma^{2}}$. Hence, we proved the essential smoothness condition. 
The conclusion follows from the G\"{a}rtner-Ellis theorem, see Theorem \ref{GEThm} in the Appendix.
\end{proof}
\begin{proof}[Proof for Proposition~\ref{prop:2}]

The result follows from the Theorem~\ref{ThmOTCSqrt} and the
G\"{a}rtner-Ellis theorem. According to this result the rate function is
given by the Legendre transform of the cumulant function 
\begin{equation}\label{Legendre}
\mathcal{I}(K,S_0) = \mbox{sup}_{\theta\in\mathbb{R}}\{\theta K - \Lambda(\theta)\},
\end{equation}
where the cumulant function $\Lambda(\theta)$ is given by 
Theorem~\ref{ThmOTCSqrt}.

(i) $K \geq S_0$. This case corresponds to 
$0 \leq \theta \leq \frac{\pi^2}{2\sigma^2}$. The cumulant function
$\Lambda(\theta)$ is given by
\begin{equation}\label{Lamdef}
\Lambda(\theta) = \frac{S_0}{\sigma^2} \sqrt{2\theta \sigma^2} 
\tan \sqrt{\frac12 \sigma^2 \theta} = 
 \frac{S_0}{\sigma^2} F_+(\theta \sigma^2)
\,.
\end{equation}
where we defined $F_+(y) := \sqrt{2y} \tan\sqrt{\frac12 y}$.

The optimal value of $\theta$ in (\ref{Legendre}) is given by the solution of 
the equation
\begin{equation}
K = S_0 F'_+(\theta_* \sigma^2)\,,
\end{equation}
with
\begin{equation}
F'_+(y) = \frac{1}{2\cos^2 \sqrt{y/2}} 
\left(1 + \frac{\sin\sqrt{2y}}{\sqrt{2y}} \right)\,.
\end{equation}
Numerical evaluation shows that $F'_+(y):[0,\infty) \to [1,\infty)$ is
a bijective map, such that this
equation will have a solution for $K>S_0$. 
Identifying 
\begin{equation}
x = \sqrt{\frac12 \theta_* \sigma^2} \,,
\end{equation}
it is easy to see that the equation for $\theta_*$ is the same as (\ref{21}).
The result for the rate function is
\begin{eqnarray}
\mathcal{I}(K,S_0) &=& \theta_* K - \Lambda(\theta_*) = 
\frac{S_0}{\sigma^2}
\left( \theta_* \sigma^2 \frac{K}{S_0} - 
F_+(\theta_* \sigma^2) \right) \\
&=& 
\frac{S_0}{\sigma^2}
\left( 2x^2 \frac{1}{2\cos^2 x} \left(1 + \frac{\sin 2x}{2x} \right) - 
2x \tan x \right) \nonumber \\
&=& \frac{S_0}{\sigma^2}
\frac{x^2}{\cos^2 x}
\left(1 - \frac{\sin (2x)}{2x} \right)\,, \nonumber
\end{eqnarray}
which yields equation (\ref{20}).

(ii) $K \leq S_0$. This case corresponds to $\theta \leq 0$. 
The cumulant function $\Lambda(\theta)$  is
\begin{equation}
\Lambda(\theta) = -\frac{S_0}{\sigma^2} 
\sqrt{-2 \theta \sigma^2} \tanh \sqrt{- \frac12 \theta \sigma^2}
= \frac{S_0}{\sigma^2} F_-(\theta \sigma^2)\,,
\end{equation}
where we introduced  $F_-(y) := -\sqrt{-2 y} \tanh\sqrt{- \frac12 y}$.
This is related to the function appearing for the previous case
as $F_-(iy) =  F_+(y)$.

The optimal $\theta$ is given by the solution of the equation
\begin{equation}\label{thetaeq2}
\frac{K}{S_0} = F'_-(\theta_*\sigma^2)\,,
\end{equation}
where
\begin{equation}
F'_-(y)= \frac{1}{2 \cosh^2 \sqrt{-\frac12 y}}
\left( 1 + \frac{\sinh \sqrt{-2y}}{\sqrt{-2y}} \right) \,.
\end{equation}
Numerical evaluation gives that $F'_-(y):(-\infty,0] \to (0,1]$ is a bijective
function,
so this equation will have a solution for $K < S_0$. 
Identifying 
\begin{equation}
x = \sqrt{-\frac12 \theta_* \sigma^2}\,.
\end{equation}
we see that the equation (\ref{thetaeq2}) reproduces (\ref{19}).
The result for the rate function is
\begin{eqnarray}
\mathcal{I}(K,S_0) &=& \theta_* K - \Lambda(\theta_*) = 
\frac{S_0}{\sigma^2}
\left( \theta_* \sigma^2 \frac{K}{S_0} - 
F_-(\theta_* \sigma^2) \right) \\
&=& 
\frac{S_0}{\sigma^2}
\left(- 2x^2 \frac{1}{2\cosh^2 x} \left(1 + \frac{\sinh 2x}{2x} \right) +
2x \tanh x \right) \nonumber \\
&=& -\frac{S_0}{\sigma^2}
\frac{x^2}{\cosh^2 x}
\left(1 - \frac{\sinh (2x)}{2x} \right)\,,
\nonumber
\end{eqnarray}
which gives the result of equation (\ref{18}).

\end{proof}


\begin{proof}[Proof of Proposition~\ref{prop:4}]

(i) This is obtained starting with the relation
\begin{equation}
\mathcal{I}(K,S_{0})=\sup_{0\leq\theta<\frac{\pi^{2}}{2\sigma^{2}}}
\left\{\theta K-\frac{\sqrt{2\theta}}{\sigma}\tan\left(\frac{\sigma}{2}\sqrt{2\theta}\right)S_{0}\right\}.
\end{equation}
On the one hand, 
$\mathcal{I}(K,S_{0})\leq\sup_{0\leq\theta<\frac{\pi^{2}}{2\sigma^{2}}}\theta K=\frac{\pi^{2}}{2\sigma^{2}}K$.
On the other hand, for any $\epsilon>0$, for sufficiently large $K$, 
\begin{equation}
\mathcal{I}(K,S_{0})=\sup_{\frac{\pi^{2}}{2\sigma^{2}}-\epsilon\leq\theta<\frac{\pi^{2}}{2\sigma^{2}}}
\left\{\theta K-\frac{\sqrt{2\theta}}{\sigma}\tan\left(\frac{\sigma}{2}\sqrt{2\theta}\right)S_{0}\right\}
\geq\left(\frac{\pi^{2}}{2\sigma^{2}}-\epsilon\right)K-\Lambda\left(\frac{\pi^{2}}{2\sigma^{2}}-\epsilon\right).
\end{equation}
Thus, $\liminf_{K\rightarrow\infty}\frac{\mathcal{I}(K,S_{0})}{K}\geq\left(\frac{\pi^{2}}{2\sigma^{2}}-\epsilon\right)$. 
Since it holds for any $\epsilon>0$, we conclude that the relation (\ref{LargeKI})
holds.

(ii) This is obtained starting from the relation
\begin{equation}
\mathcal{I}(K,S_{0})
=\sup_{\theta\leq 0}\left\{K\theta+\frac{\sqrt{-2\theta}}{\sigma}
\tanh\left(\frac{\sigma}{2}\sqrt{-2\theta}\right)S_{0}\right\}\,.
\end{equation}
At optimality we have
\begin{equation}
K=\frac{\sqrt{2}}{2\sigma\sqrt{-\theta}}\tanh\left(\frac{\sigma}{2}\sqrt{-2\theta}\right)S_{0}
+\frac{1}{2}\left[1-\tanh^{2}\left(\frac{\sigma}{2}\sqrt{-2\theta}\right)\right]S_{0}.
\end{equation}
Note that the function $\tanh x$ approaches $1$ exponentially fast as 
$x\rightarrow\infty$.
Therefore, $\theta\sim-\frac{S_{0}^{2}}{2\sigma^{2}K^{2}}$ as 
$K\rightarrow 0$ and the result (\ref{rem4eq}) follows.
\end{proof}


\subsection{Proofs of the results in Section~\ref{Sec:3}}

\begin{proof}[Proof of Lemma~\ref{lemma:1}]
We will prove the result for the case of the Asian call option.
The case of the Asian put option is very similar.

Note that by H\"{o}lder's inequality, for any $\frac{1}{p}+\frac{1}{p'}=1$, 
$p,p'>1$ and $p\geq 2$,
\begin{align}
C(T)&= e^{-rT}\mathbb{E}\left[\left|\frac{1}{T}\int_{0}^{T}S_{t}dt-K\right|1_{\frac{1}{T}\int_{0}^{T}S_{t}dt\geq K}\right]
\\
&\leq e^{-rT}\left(\mathbb{E}\left[\left|\frac{1}{T}\int_{0}^{T}S_{t}dt-K\right|^{p}\right]\right)^{\frac{1}{p}}
\mathbb{P}\left(\frac{1}{T}\int_{0}^{T}S_{t}dt\geq K\right)^{\frac{1}{p'}}
\nonumber
\\
&\leq
e^{-rT} 2^{\frac{p-1}{p}}
\left(K^p + \left(\mathbb{E}\left[\frac{1}{T}\int_{0}^{T}S_{t}^{p}dt\right]\right)\right)^{1/p}
\mathbb{P}\left(\frac{1}{T}\int_{0}^{T}S_{t}dt\geq K\right)^{\frac{1}{p'}}\,,
\nonumber
\end{align}
where in the last step we used Jensen's inequality to write
\begin{eqnarray}
&& \mathbb{E}\left[ \left| \frac{1}{T} \int_0^T S_t dt - K \right|^p \right] \leq
\mathbb{E}\left[ \left( \frac{1}{T}\int_0^T S_t dt + K \right)^p \right] \\
&& \qquad \leq 2^{p-1} 
\mathbb{E}\left[ \left( \frac{1}{T}\int_0^T S_t dt \right)^p + K^p\right] 
\leq 2^{p-1} 
\mathbb{E}\left[ \left( \frac{1}{T}\int_0^T S_t^{p} dt \right) + K^p\right]\,.\nonumber
\end{eqnarray}

The second inequality follows by noting that for $p\geq 2$, $x \to x^p$ is a convex
function for $x\geq 0$, which gives by Jensen's inequality $\left( \frac{x+y}{2}\right)^p
\leq \frac{x^p+y^p}{2}$ for any $x,y\geq 0$. This gives
\begin{align}\label{UpI}
\mathbb{E}\left[\left|\frac{1}{T}\int_{0}^{T}S_{t}dt-K\right|^{p}\right]
&\leq\mathbb{E}\left[\left(\frac{1}{T}\int_{0}^{T}S_{t}dt+K\right)^{p}\right]
\\
&\leq 2^{p-1}\left[\mathbb{E}
\left[\left(\frac{1}{T}\int_{0}^{T}S_{t}dt\right)^{p}\right]+K^{p}\right].
\nonumber
\end{align}
The last inequality follows again from the Jensen's inequality which gives for $p\geq 2$
$\mathbb{E}[(\frac{1}{T}\int_0^T S_t dt )^{p}] \leq \mathbb{E}[\frac{1}{T} \int_0^T S_t^p dt]$.

For any $p\geq 2$, 
\begin{equation}
\frac{1}{T}\int_{0}^{T}\mathbb{E}[S_{t}^{p}]dt=O(1)\,,
\end{equation}
since for the CEV process, all these moments are finite and well-behaved as $T\to 0$.
The marginal distribution of $S_{t}$ in this model is known \cite{LinetskyMendoza}
and the above expression can be computed explicitly.

Therefore, we have
\begin{equation}
\limsup_{T\rightarrow 0}T\log C(T)
\leq\limsup_{T\rightarrow 0}\frac{1}{p'}T\log\mathbb{P}\left(\frac{1}{T}\int_{0}^{T}S_{t}dt\geq K\right).
\end{equation}
Since it holds for any $2>p'>1$, we have the upper bound.

Next we derive a matching lower bound on $C(T)$. For any $\epsilon>0$,
\begin{align}
C(T)&\geq e^{-rT}\mathbb{E}\left[\left(\frac{1}{T}\int_{0}^{T}S_{t}dt-K\right)1_{\frac{1}{T}\int_{0}^{T}S_{t}dt\geq K+\epsilon}\right]
\\
&\geq e^{-rT}\epsilon\mathbb{P}\left(\frac{1}{T}\int_{0}^{T}S_{t}dt\geq K+\epsilon\right),
\nonumber
\end{align}
which implies that
\begin{equation}
\liminf_{T\rightarrow 0}T\log C(T)
\geq\liminf_{T\rightarrow 0}T\log\mathbb{P}\left(\frac{1}{T}\int_{0}^{T}S_{t}dt\geq K+\epsilon\right).
\end{equation}
Since it holds for any $\epsilon>0$, we get the lower bound by letting $\epsilon\rightarrow 0$, provided that
the limit $\mathcal{I}(K,S_{0}):=-\lim_{T\rightarrow 0}T\log\mathbb{P}\left(\frac{1}{T}\int_{0}^{T}S_{t}dt\geq K\right)$
exists and is continuous in $K$. The continuity in $K$ can be seen from the expression in Proposition~\ref{prop:IKalt}.
\end{proof}


\begin{proof}[Proof of Theorem~\ref{ThmOTMCEV}]

We split the proof into several steps.

\textbf{Step 1}. We need to prove that
\begin{equation}\label{S1}
\lim_{T\rightarrow 0}T\log\mathbb{P}\left(\frac{1}{T}\int_{0}^{T}S_{t}dt\geq K\right)
=\lim_{T\rightarrow 0}T\log\mathbb{P}\left(\frac{1}{T}\int_{0}^{T}\hat{S}_{t}dt\geq K\right)\,,
\end{equation}
where 
\begin{equation}\label{hatS}
d\hat{S}_{t}=\sigma\hat{S}_{t}^{\beta}dW_{t},
\end{equation}
with $\hat{S}_{0}=S_{0}$. That is, the drift term is negligible for small time large deviations.
Let us now prove \eqref{S1}. Note that
\begin{equation}
S_{t}=S_{0}e^{(r-q)t+\int_{0}^{t}\sigma S_{s}^{\beta}dW_{s}-\frac{1}{2}\sigma^{2}\int_{0}^{t}S_{s}^{2\beta}ds}
=e^{(r-q)t}\tilde{S}_{t},
\end{equation}
where
\begin{equation}
d\tilde{S}_{t}=\sigma\tilde{S}_{t}^{\beta}e^{-(r-q)\beta t}dW_{t},
\qquad
\tilde{S}_{0}=S_{0}>0.
\end{equation}
By the time change $d\tau(t)=e^{-2(r-q)\beta t}dt$, $\tau(0)=0$, 
$\tilde{S}_{t}=\hat{S}_{\tau(t)}$, where $\hat{S}$ is defined in \eqref{hatS}.

Hence, 
\begin{align}
&\lim_{T\rightarrow 0}T\log\mathbb{P}\left(\frac{1}{T}\int_{0}^{T}S_{t}dt\geq K\right)
\\
&=\lim_{T\rightarrow 0}T\log\mathbb{P}\left(\frac{1}{T}\int_{0}^{T}e^{(r-q)t}\hat{S}_{\tau(t)}dt\geq K\right)
\nonumber
\\
&=\lim_{T\rightarrow 0}T\log\mathbb{P}\left(\frac{1}{T}\int_{0}^{\tau(T)}e^{(r-q)(1-2\beta)\tau^{-1}(t)}\hat{S}_{t}dt\geq K\right)
\,.
\nonumber
\end{align}
It is easy to check that $\frac{\tau(T)}{T}\rightarrow 1$ as $T\rightarrow 0$
and $\lim_{T\rightarrow 0}\inf_{0\leq t\leq T}e^{(r-q)(1-2\beta)\tau^{-1}(t)}
=\lim_{T\rightarrow 0}\sup_{0\leq t\leq T}e^{(r-q)(1-2\beta)\tau^{-1}(t)}=1$. 
Hence, \eqref{S1} follows.

\textbf{Step 2}. Now assume that $r=q=0$ so that
\begin{equation}
dS_{t}=\sigma S_{t}^{\beta}dW_{t},
\end{equation}
with $S_{0}>0$. Therefore, for $0\leq t\leq 1$, 
\begin{equation}
dS_{tT}=\sigma S_{tT}^{\beta}dW_{tT}
=\sqrt{T}\sigma S_{tT}^{\beta}d(W_{tT}/\sqrt{T})
=\sqrt{T}\sigma S_{tT}^{\beta}dB_{t},
\end{equation}
where $B_{t}:=W_{tT}/\sqrt{T}$ is a standard Brownian motion by the scaling property of
the Brownian motion. Therefore, by letting $T=\epsilon$, 
\begin{equation}
\lim_{T\rightarrow 0}T\log\mathbb{P}\left(\int_{0}^{1}S_{tT}dt\geq K\right)
=\lim_{\epsilon\rightarrow 0}\epsilon\log\mathbb{P}\left(\int_{0}^{1}S_{t}^{\epsilon}dt\geq K\right),
\end{equation}
where
\begin{equation}
dS_{t}^{\epsilon}=\sqrt{\epsilon}\sigma(S_{t}^{\epsilon})^{\beta}dB_{t},
\end{equation}
with $S_{0}^{\epsilon}=S_{0}>0$. 

\textbf{Step 3}. We need to show that 
\begin{equation}\label{S3}
\lim_{\epsilon\rightarrow 0}
\epsilon\log\mathbb{P}\left(\int_{0}^{1}S_{t}^{\epsilon}dt\geq K\right)
=\lim_{\delta\rightarrow 0}
\lim_{\epsilon\rightarrow 0}
\epsilon\log\mathbb{P}\left(\int_{0}^{1}S_{t}^{\epsilon}dt\geq K, S_{t}^{\epsilon}\geq\delta,0\leq t\leq 1\right).
\end{equation}
Note that conditional of $\int_{0}^{1}S_{t}^{\epsilon}dt\geq K$, 
the event that $S_{t}^{\epsilon}\geq\delta,0\leq t\leq 1$ is a typical event, while
the event that $S_{t}^{\epsilon}\leq\delta$ for some $0\leq t\leq 1$ is a rare event. 
Therefore, for sufficiently small $\delta>0$, 
\begin{equation}\label{S3I}
\mathbb{P}\left(\int_{0}^{1}S_{t}^{\epsilon}dt\geq K\right)
\leq 2\mathbb{P}\left(\int_{0}^{1}S_{t}^{\epsilon}dt\geq K, S_{t}^{\epsilon}\geq\delta,0\leq t\leq 1\right).
\end{equation}
On the other hand, for any $\delta>0$,
\begin{equation}
\mathbb{P}\left(\int_{0}^{1}S_{t}^{\epsilon}dt\geq K\right)
\geq\mathbb{P}\left(\int_{0}^{1}S_{t}^{\epsilon}dt\geq K, S_{t}^{\epsilon}\geq\delta,0\leq t\leq 1\right),
\end{equation}
which implies that, for any $\delta>0$. 
\begin{equation}\label{S3II}
\lim_{\epsilon\rightarrow 0}
\epsilon\log\mathbb{P}\left(\int_{0}^{1}S_{t}^{\epsilon}dt\geq K\right)
\geq
\lim_{\epsilon\rightarrow 0}
\epsilon\log\mathbb{P}\left(\int_{0}^{1}S_{t}^{\epsilon}dt\geq K, S_{t}^{\epsilon}\geq\delta,0\leq t\leq 1\right).
\end{equation}
Hence, \eqref{S3} follows from \eqref{S3I} and \eqref{S3II}.

\textbf{Step 4}. Define
\begin{equation}\label{Sed}
dS_{t}^{\epsilon,\delta}=b^{\delta}(S_{t}^{\epsilon,\delta})dt+
\sqrt{\epsilon}\sigma(S_{t}^{\epsilon,\delta})^{\beta}dB_{t}\,,\quad
S_0^{\epsilon,\delta}=S_0\,,
\end{equation}
where $b^{\delta}(x)=0$ for any $x>\delta$ and also is locally Lipschitz continuous and $b^{\delta}(0)>0$.
Morever, $S\mapsto S^{\beta}$ is H\"{o}lder continuous with exponent $\geq\frac{1}{2}$
and for $\beta<1$, it has sublinear growth at $\infty$. 
The dynamics \eqref{Sed} satisfies the assumption A1.1. in Baldi and Caramellino \cite{BC2011}.
It is easy to see that
\begin{equation}
\mathbb{P}\left(\int_{0}^{1}S_{t}^{\epsilon}dt\geq K, S_{t}^{\epsilon}\geq\delta,0\leq t\leq 1\right)
=\mathbb{P}\left(\int_{0}^{1}S_{t}^{\epsilon,\delta}dt\geq K, S_{t}^{\epsilon,\delta}\geq\delta,0\leq t\leq 1\right).
\end{equation}
By Theorem 1.2 in Baldi and Caramellino \cite{BC2011} it follows that
$\mathbb{P}(S^{\epsilon,\delta}\in\cdot)$
satisfies a large deviation principle on $C_{S_{0}}([0,1])$, 
the space of continuous functions starting at $S_{0}$ equipped with uniform topology, 
with the rate function $\frac{1}{2}\int_{0}^{1}\frac{(g'(t)-b^{\delta}(g(t)))^{2}}{\sigma^{2}g(t)^{2\beta}}dt$, 
with the understanding that the rate function is $+\infty$ if $g$ is not differentiable.
Moreover, the map $g\mapsto(\int_{0}^{1}g(t)dt,g)$ is continuous
from $C_{S_{0}}[0,1]$ to $\mathbb{R}_{+}\times C_{S_{0}}[0,1]$.

By the contraction principle, see Theorem~\ref{Contraction}
in the Appendix, we have
\begin{align}
&\lim_{\epsilon\rightarrow 0}
\epsilon\log\mathbb{P}\left(\int_{0}^{1}S_{t}^{\epsilon}dt\geq K, S_{t}^{\epsilon,\delta}\geq\delta,0\leq t\leq 1\right)
\\
&=-\inf_{\int_{0}^{1}g(t)dt\geq K, g(0)=S_{0}, g(t)\geq\delta, 0\leq t\leq 1}
\frac{1}{2}\int_{0}^{1}\frac{(g'(t)-b^{\delta}(g(t)))^{2}}{\sigma^{2}g(t)^{2\beta}}dt
\nonumber
\\
&=-\inf_{\int_{0}^{1}g(t)dt\geq K, g(0)=S_{0}, g(t)\geq\delta, 0\leq t\leq 1}
\frac{1}{2}\int_{0}^{1}\frac{(g'(t))^{2}}{\sigma^{2}g(t)^{2\beta}}dt.
\nonumber
\end{align}
Thus,
\begin{align}
&\lim_{\epsilon\rightarrow 0}
\epsilon\log\mathbb{P}\left(\int_{0}^{1}S_{t}^{\epsilon}dt\geq K\right)
\\
&=\lim_{\delta\rightarrow 0}\lim_{\epsilon\rightarrow 0}
\epsilon\log\mathbb{P}\left(\int_{0}^{1}S_{t}^{\epsilon}dt\geq K, S_{t}^{\epsilon}\geq\delta,0\leq t\leq 1\right)
\nonumber \\
&=-\inf_{\int_{0}^{1}g(t)dt\geq K, g(0)=S_{0}, g(t)\geq 0, 0\leq t\leq 1}
\frac{1}{2}\int_{0}^{1}\frac{(g'(t))^{2}}{\sigma^{2}g(t)^{2\beta}}dt.
\nonumber
\end{align}
\end{proof}


\begin{proof}[Proof of Theorem~\ref{ThmATMSqrt}]
We will only prove the case for the call option here.
The proof for the put option is very similar and hence omitted.
As $T\rightarrow 0$,
\begin{equation}
C(T)=e^{-rT}\mathbb{E}\left[\left(\frac{1}{T}\int_{0}^{T}S_{t}dt-K\right)^{+}\right]
=\mathbb{E}\left[\left(\frac{1}{T}\int_{0}^{T}S_{t}dt-K\right)^{+}\right]+O(T),
\end{equation}
and we showed that
\begin{equation}
\mathbb{E}\left[\left(\frac{1}{T}\int_{0}^{T}S_{t}dt-K\right)^{+}\right]
=\mathbb{E}\left[\left(\frac{1}{T}\int_{0}^{\tau(T)}e^{(r-q)(1-2\beta)\tau^{-1}(t)}\hat{S}_{t}dt-K\right)^{+}\right],
\end{equation}
where $d\hat{S}_{t}=\sigma\hat{S}_{t}^{\beta}dW_{t}$ and $\hat{S}_{0}=S_{0}$.

It is easy to show that
\begin{align}
&\left|\mathbb{E}\left[\left(\frac{1}{T}\int_{0}^{\tau(T)}e^{(r-q)(1-2\beta)\tau^{-1}(t)}\hat{S}_{t}dt-K\right)^{+}\right]
-\mathbb{E}\left[\left(\frac{1}{T}\int_{0}^{\tau(T)}\hat{S}_{t}dt-K\right)^{+}\right]\right|
\\
&\leq\mathbb{E}\left[\frac{1}{T}\int_{0}^{\tau(T)}|e^{(r-q)(1-2\beta)\tau^{-1}(t)}-1|\hat{S}_{t}dt\right]
\nonumber
\\
&=S_{0}\frac{1}{T}\int_{0}^{\tau(T)}|e^{(r-q)(1-2\beta)\tau^{-1}(t)}-1|dt=O(T).
\nonumber
\end{align}
Moreover, we can show that
\begin{align}
&\left|\mathbb{E}\left[\left(\frac{1}{T}\int_{0}^{T}\hat{S}_{t}dt-K\right)^{+}\right]
-\mathbb{E}\left[\left(\frac{1}{T}\int_{0}^{\tau(T)}\hat{S}_{t}dt-K\right)^{+}\right]\right|
\\
&\leq\mathbb{E}\left|\frac{1}{T}\int_{\tau(T)}^{T}\hat{S}_{t}dt\right|
=S_{0}\frac{1}{T}|T-\tau(T)|=O(T).
\nonumber
\end{align}

Next, let $dX_{t}=\sigma S_{0}^{\beta}dW_{t}$ and $X_{0}=S_{0}$, that is $X_{t}=S_{0}+\sigma S_{0}^{\beta}W_{t}$.
By It\^{o}'s formula and taking the expectations, we get
\begin{align}
\mathbb{E}(\hat{S}_{t}-X_{t})^{2}
&=\sigma^{2}\int_{0}^{t}\mathbb{E}(\hat{S}_{s}^{\beta}-S_{0}^{\beta})^{2}ds
\\
&\leq 2\sigma^{2}\int_{0}^{t}\mathbb{E}[(\hat{S}_{s}^{\beta}-X_{s}^{\beta})^{2}]ds
+2\sigma^{2}\int_{0}^{t}\mathbb{E}[(X_{s}^{\beta}-S_{0}^{\beta})^{2}]ds.
\nonumber
\end{align}
For any $x>0$, $y\geq 0$ and $\frac{1}{2}\leq\beta<1$, 
we have $|x^{\beta}-y^{\beta}|\leq|x-y|x^{\beta-1}$, 
see e.g. Lemma 2.2. in Cai and Wang \cite{CW}. Hence,
\begin{equation}
2\sigma^{2}\int_{0}^{t}\mathbb{E}[(X_{s}^{\beta}-S_{0}^{\beta})^{2}]ds
\leq 2\sigma^{2}S_{0}^{2(\beta-1)}\int_{0}^{t}\mathbb{E}[(X_{s}-S_{0})^{2}]ds
=\sigma^{2}S_{0}^{2(\beta-1)}\sigma^{2}S_{0}^{2\beta}t^{2}.
\end{equation}
Moreover, for $S_{0}>\delta>0$,
\begin{align}
&2\sigma^{2}\int_{0}^{t}\mathbb{E}[(\hat{S}_{s}^{\beta}-X_{s}^{\beta})^{2}]ds
\\
&=2\sigma^{2}\int_{0}^{t}\mathbb{E}[(\hat{S}_{s}^{\beta}-X_{s}^{\beta})^{2}1_{X_{s}\geq\delta}]ds
+2\sigma^{2}\int_{0}^{t}\mathbb{E}[(\hat{S}_{s}^{\beta}-X_{s}^{\beta})^{2}1_{X_{s}<\delta}]ds.
\nonumber
\end{align}
On the one hand,
\begin{equation}
2\sigma^{2}\int_{0}^{t}\mathbb{E}[(\hat{S}_{s}^{\beta}-X_{s}^{\beta})^{2}1_{X_{s}\geq\delta}]ds
\leq 2\sigma^{2}\delta^{2(\beta-1)}\int_{0}^{t}\mathbb{E}[(\hat{S}_{s}-X_{s})^{2}]ds.
\end{equation}
On the other hand,
\begin{align}
&2\sigma^{2}\int_{0}^{t}\mathbb{E}[(\hat{S}_{s}^{\beta}-X_{s}^{\beta})^{2}1_{X_{s}<\delta}]ds
\\
&\leq 2\sigma^{2}\int_{0}^{t}\sqrt{\mathbb{E}[(\hat{S}_{s}^{\beta}-X_{s}^{\beta})^{4}]}\sqrt{\mathbb{P}(X_{s}<\delta)}ds
\nonumber
\\
&\leq 2\sigma^{2}\max_{0\leq s\leq t}\sqrt{\mathbb{P}(X_{s}<\delta)}
\int_{0}^{t}\sqrt{\mathbb{E}[(\hat{S}_{s}^{\beta}-X_{s}^{\beta})^{4}]}ds.
\nonumber
\end{align}
Note that 
\begin{equation}
\int_{0}^{t}\sqrt{\mathbb{E}[(\hat{S}_{s}^{\beta}-X_{s}^{\beta})^{4}]}ds
\leq\int_{0}^{t}\sqrt{4\mathbb{E}[\hat{S}_{s}^{4\beta}+X_{s}^{4\beta}]}ds,
\end{equation}
and we can compute $\mathbb{E}[\hat{S}_{s}^{4\beta}]$ and $\mathbb{E}[X_{s}^{4\beta}]$ explicitly
since $\hat{S}_{t}$ is a CEV process and $X_{t}$ is a Brownian motion. 
It is therefore easy to check that $\int_{0}^{T}\sqrt{\mathbb{E}[(\hat{S}_{s}^{\beta}-X_{s}^{\beta})^{4}]}ds=O(T)$.
Furthermore, 
\begin{equation}
2\sigma^{2}\max_{0\leq s\leq t}\sqrt{\mathbb{P}(X_{s}<\delta)}
=2\sigma^{2}\Phi\left(\frac{\delta-S_{0}}{\sigma S_{0}^{\beta}\sqrt{t}}\right),
\end{equation}
where $\Phi(x):=\frac{1}{\sqrt{2\pi}}\int_{-\infty}^{x}e^{-\frac{y^{2}}{2}}dy$. 
Hence, by Gronwall's inequality, we conclude that 
\begin{equation}
\mathbb{E}[(\hat{S}_{T}-X_{T})^{2}]=O(T^{2}).
\end{equation}
Note that $\hat{S}_{t}-X_{t}$ is a martingale. By Doob's martingale inequality,
\begin{equation}
\mathbb{E}\left[\max_{0\leq t\leq T}|\hat{S}_{t}-X_{t}|\right]
\leq C\sqrt{\mathbb{E}[(\hat{S}_{T}-X_{T})^{2}]}=O(T).
\end{equation}
Therefore, we conclude that 
\begin{align}
C(T)&=\mathbb{E}\left[\left(\frac{1}{T}\int_{0}^{T}X_{t}dt-S_{0}\right)^{+}\right]+O(T)
\\
&=\mathbb{E}\left[\left(\sigma S_{0}^{\beta}\frac{1}{T}\int_{0}^{T}W_{t}dt\right)^{+}\right]+O(T)
\nonumber
\\
&=\sigma S_{0}^{\beta}\frac{\sqrt{T}}{\sqrt{3}}\mathbb{E}[Z1_{Z>0}]+O(T),
\nonumber
\end{align}
where $Z\sim N(0,1)$. Finally, we can compute that
\begin{equation}
\mathbb{E}[Z1_{Z>0}]=\frac{1}{\sqrt{2\pi}}\int_{0}^{\infty}xe^{-\frac{x^{2}}{2}}dx
=\frac{1}{\sqrt{2\pi}}.
\end{equation}
Hence, we proved the desired result.
\end{proof}


\begin{proof}[Proof of Proposition~\ref{VarProp}]

We will define $\mathcal{I}_K(K,S_0)$ as the solution of the variational 
problem (\ref{Ivarproblem}), obtained by replacing the inequality 
(\ref{constineq}) with the equality constraint $\int_0^1 g(t) dt = K$. 
This is solved by 
considering the variational problem for the auxiliary functional
\begin{equation}
\Lambda[g] := \frac{1}{2\sigma^2}
\int_0^1  \frac{(g'(t))^2}{g(t)^{2\beta}} dt - \lambda
\left( \int_0^1 g(t) dt - K \right)\,,
\end{equation}
where $\lambda$ is a Lagrange multiplier. 

The solution of this variational problem satisfies the Euler-Lagrange equation
\begin{equation}\label{ELg}
g''(t) = \beta \frac{[g'(t)]^2}{g(t)} - \lambda \sigma^2 (g(t))^{2\beta}\,,
\end{equation}
with initial condition $g(0)=S_0$ and transversality condition $g'(1)=0$.

This equation can be simplified by the change of variable
\begin{equation}
g(t) = S_0 ( y(t) )^{\frac{1}{1-\beta}}\,.
\end{equation}
Expressed in terms of $y(t)$, the Euler-Lagrange equation (\ref{ELg}) becomes
\begin{equation}\label{EmdenFowler}
y''(t) = C (y(t))^{\frac{\beta}{1-\beta}} \,,
\end{equation}
with $C := - \lambda \sigma^2 (1-\beta) S_0^{2\beta-1}$. The solution $y(t)$
satisfies the initial condition $y(0)=1$ and transversality condition
$y'(1)=0$. The rate function is expressed in terms of this solution as
\begin{equation}\label{Iratey}
\mathcal{I}_K(K,S_0) = \frac{S_0^{2-2\beta}}{2\sigma^2 (1-\beta)^2}
\int_0^1 [y'(t)]^2 dt \,.
\end{equation}
The constraint $\int_0^1 g(t) dt = K$ reads 
\begin{equation}\label{constrainty}
\int_0^1 (y(t))^{\frac{1}{1-\beta}} dt = \frac{K}{S_0}\,.
\end{equation}

The differential equation (\ref{EmdenFowler}) is known as the Emden-Fowler
equation. The exponent $\gamma :=\frac{\beta}{1-\beta}$ satisfies $\gamma\geq 1$
for the cases considered here $\beta \in [\frac12, 1)$. This equation can
be reduced to a first order ODE by noting the conservation of the quantity
\begin{equation}
E := \frac12 [y'(t)]^2 - C (1-\beta) (y(t))^{\gamma+1}\,.
\end{equation}

Taking into account the boundary condition $y'(1)=0$ we get the relation
\begin{equation}\label{yprime}
[y'(t)]^2 = 2C(1-\beta)\left([y(t)]^{\gamma+1} - y_1^{\gamma+1}\right)\,,
\end{equation}
where we denoted $y_1 := y(1)$.

We distinguish the two cases:

1. $C>0$. This corresponds to $y'(t)<0$ and $y(1) < y(0)=1$. 
From (\ref{constrainty}) we get that this corresponds to $K < S_0$.

2. $C<0$. This corresponds to $y'(t)>0$ and $y(1) > y(0)=1$. 
From (\ref{constrainty}) we get that this corresponds to $K > S_0$.

We consider the two cases separately.

{\bf Case 1.} $C>0$. We can express $y_1$ in terms of $C$ using the relation
\begin{equation}
1 = \int_0^1 dt =  \int_{y_1}^{y(0)} \frac{dy}{y'} =
\frac{1}{\sqrt{2C(1-\beta)}} 
\int_{y_1}^1 \frac{dy}{\sqrt{y^{\gamma+1} - y_1^{\gamma+1}}}.
\end{equation}
This relation can be used to eliminate $C$ in terms of $y_1$ as
\begin{equation}
C = \frac{1}{2(1-\beta)} [A^{(+)}(y_1)]^2 \,,
\end{equation}
where we defined the function
\begin{align}
A^{(+)}(x) 
&:= \int_{x}^1 \frac{dy}{\sqrt{y^{\gamma+1} - x^{\gamma+1}}} 
\\
&= \frac{2x}{\gamma+1} \frac{\sqrt{1-x^{\gamma+1}}}{x^{\gamma+1}}
{}_2F_1\left( \frac{\gamma}{\gamma+1}, \frac12; \frac32; 1 - \frac{1}{x^{\gamma+1}} \right) 
\,,\quad 0 < x \leq 1\,. \nonumber
\end{align}

The constraint (\ref{constrainty}) can be written equivalently using 
(\ref{yprime}) as
\begin{equation}
\frac{K}{S_0} = \int_0^1 [y(t)]^{\gamma+1} dt =
[y_1]^{\gamma+1} + \frac{1}{2C(1-\beta)} \int_0^1 [y'(t)]^2 dy\,.
\end{equation}
The integral can be expressed by a change of variable as
\begin{eqnarray}\label{139}
\int_0^1 dy [y'(t)]^2 &=& \int_{y(0)}^{y(1)} y'dy = 
\sqrt{2C(1-\beta)} \int_{y(1)}^{1} \sqrt{y^{\gamma+1} - y_1^{\gamma+1}}dy\\
 &=& A^{(+)}(y(1)) B^{(+)}(y(1))\,, \nonumber
\end{eqnarray}
where we defined
\begin{align}
B^{(+)}(x) 
&:= \int_{x}^1 \sqrt{y^{\gamma+1} - x^{\gamma+1}} dy
\\
&= \frac{2x}{3(\gamma+1)} \frac{(1-x^{\gamma+1})^{3/2}}{x^{\gamma+1}}
{}_2F_1\left( \frac{\gamma}{\gamma+1}, \frac32; \frac52; 1 - \frac{1}{x^{\gamma+1}} \right)
\,,\quad 0 < x \leq 1\,.\nonumber
\end{align}
The integral (\ref{139}) is the same as the integral appearing in the 
expression for the rate function (\ref{Iratey}).

In conclusion, the rate function $\mathcal{I}_K(K,S_0)$ for $K<S_0$ 
is given by
\begin{equation}\label{IKsol2}
\mathcal{I}_K(K,S_0) = \frac{S_0^{2(1-\beta)}}{2\sigma^2 (1-\beta)^2}
A^{(+)}(y_1) B^{(+)}(y_1)\,,
\end{equation}
where $y_1 < 1$ is the solution of the equation
\begin{equation}\label{y1eq2}
\frac{K}{S_0} = 
y_1^{\gamma+1} + \frac{B^{(+)}(y_1)}{A^{(+)}(y_1)} \,.
\end{equation}

{\bf Case 2.} $C<0$.
We can express $y(1)$ in terms of $C$ using the relation
\begin{equation}
1 = \int_0^1 dt =  \int_{y(0)}^{y(1)} \frac{dy}{y'} =
\frac{1}{\sqrt{-2C(1-\beta)}} 
\int_1^{y_1} \frac{dy}{\sqrt{y_1^{\gamma+1} - y^{\gamma+1}}}.
\end{equation}
We can use this relation to eliminate $-C>0$ in terms of $y_1$ as
\begin{equation}
-C = \frac{1}{2(1-\beta)} [A^{(-)}(y_1)]^2 \,,
\end{equation}
where we defined the function
\begin{align}
A^{(-)}(x) &:= \int_1^{x} \frac{dy}{\sqrt{x^{\gamma+1} - y^{\gamma+1}}} 
\\
&= \frac{2x}{\gamma+1} \frac{\sqrt{x^{\gamma+1}-1}}{x^{\gamma+1}}
{}_2F_1\left( \frac{\gamma}{\gamma+1}, \frac12; \frac32; 1-\frac{1}{x^{\gamma+1}} \right) 
\,,\quad x \geq 1 \,.\nonumber
\end{align}

The constraint (\ref{constrainty}) can be written equivalently using 
(\ref{yprime}) as
\begin{equation}
\frac{K}{S_0} = \int_0^1 [y(t)]^{\gamma+1} dt=
y_1^{\gamma+1} + \frac{1}{2C(1-\beta)} \int_0^1 [y'(t)]^2 dy\,.
\end{equation}
The integral can be written by a change of variable as
\begin{eqnarray}
\int_0^1 [y'(t)]^2 dy
&=& \int_{y(0)}^{y(1)} y'dy = 
\sqrt{-2C(1-\beta)} \int_1^{y_1} \sqrt{y_1^{\gamma+1} - y^{\gamma+1}}dy
\\
&=& A^{(-)}(y_1) B^{(-)}(y_1), \nonumber
\end{eqnarray}
where we defined
\begin{align}
B^{(-)}(x) &:= \int_1^x \sqrt{x^{\gamma+1} - y^{\gamma+1}} dy
\\
&= \frac{2x}{3(\gamma+1)} \frac{(x^{\gamma+1}-1)^{3/2}}{x^{\gamma+1}}
{}_2F_1\left( \frac{\gamma}{\gamma+1}, \frac32; \frac52; 1-\frac{1}{x^{\gamma+1}}  \right)
\,,\quad x \geq 1 \,.\nonumber
\end{align}
This gives also the integral appearing in the expression for the rate function
in (\ref{Iratey}).

In conclusion, the rate function $\mathcal{I}_K(K,S_0)$
for $K > S_0$ is given by
\begin{equation}\label{IKsol1}
\mathcal{I}_K(K,S_0) = \frac{S_0^{2(1-\beta)}}{2\sigma^2 (1-\beta)^2}
A^{(-)}(y_1) B^{(-)}(y_1),
\end{equation}
where $y_1 > 1$ is the solution of the equation
\begin{equation}\label{y1eq1}
\frac{K}{S_0} = 
y_1^{\gamma+1} - \frac{B^{(-)}(y_1)}{A^{(-)}(y_1)}.
\end{equation}

The integrals $A^{(\pm)}(x), B^{(\pm)}(x)$ have been evaluated in closed form
in terms of the hypergeometric function ${}_2 F_1(a,b;c;z)$, defined as
\begin{equation}
{}_2 F_1(a,b;c;z) = \frac{\Gamma(c)}{\Gamma(b)\Gamma(c-b)}
\int_0^1 t^{b-1} \frac{(1-t)^{c-b-1}}{(1-tz)^a}dt\,.
\end{equation}

The results can be simplified by changing the variable $y_1^{\gamma+1} = z$
and introducing the functions
$a^{(\pm)}(z) := (1-\beta) A^{(\pm)}(y_1)$, and $b^{(\pm)}(z) := (1-\beta) B^{(\pm)}(y_1)$.
\end{proof}

\begin{proof}[Proof of Proposition~\ref{prop:IKalt}]

(i) The extremum condition for $\varphi$ is 
\begin{equation}\label{varphieq}
\varphi_* - \frac{K}{S_0} = \frac{\mathcal{G}^{(-)}(\varphi_*)}
{\mathcal{F}^{(-)}(\varphi_*)}\,,
\end{equation}
where 
\begin{equation}\label{Fmdef}
\mathcal{F}^{(-)}(\varphi) := 2 \frac{d}{d\varphi} 
\mathcal{G}^{(-)}(\varphi) = 
\frac{S_0^{1-\beta}}{\sigma}
\int_1^\varphi \frac{dz}{z^{\beta} \sqrt{\varphi-z}}\,.
\end{equation}

The equation (\ref{varphieq}) is identical with the equation (\ref{y1eq1})
for $y_1$, identifying $\varphi_* = y_1^{\gamma+1}$. 

Substituting (\ref{varphieq}) into (\ref{IK1alt}) we have
\begin{equation}
\mathcal{I}_K(K,S_0) = \frac12 \mathcal{F}^{(-)}(\varphi_*)
\mathcal{G}^{(-)}(\varphi_*)\,.
\end{equation}
This result is identical to (\ref{IKsol1}) with the identification
$\varphi_* = y_1^{\gamma+1}$.

(ii) The extremum condition for $\chi$ is 
\begin{equation}\label{chieq}
\frac{K}{S_0} - \chi_* = \frac{\mathcal{G}^{(+)}(\chi_*)}
{\mathcal{F}^{(+)}(\chi_*)}\,,
\end{equation}
where 
\begin{equation}\label{Fpldef}
\mathcal{F}^{(+)}(\chi) := - 2 \frac{d}{d\chi} 
\mathcal{G}^{(+)}(\chi) = 
\frac{S_0^{1-\beta}}{\sigma}
\int_{\chi}^1 \frac{dz}{z^{\beta} \sqrt{z-\chi}}\,.
\end{equation}

The equation (\ref{chieq}) is identical with the equation 
(\ref{y1eq2}) for $y_1$, identifying $\chi_* = y_1^{\gamma+1}$. 

Substituting (\ref{chieq}) into (\ref{IK2alt}) we have
\begin{equation}
\mathcal{I}_K(K,S_0) = \frac12 \mathcal{F}^{(+)}(\chi_*)
\mathcal{G}^{(+)}(\chi_*)\,.
\end{equation}
This result is identical to (\ref{IKsol2}) with the identification
$\chi_* = y_1^{\gamma+1}$.
\end{proof}


\begin{proof}[Proof of Corollary~\ref{cor:mono}]

(i) follows from Lemma 29 in \cite{ShortMatAsian}.
The technical conditions of this Lemma require that $\mathcal{G}^{(-)}(\varphi)$
is an increasing function and that $[\mathcal{G}^{(-)}(\varphi)]^2$
has superlinear growth as $\varphi\to \infty$. The first condition is satisfied
as the derivative of $\mathcal{G}^{(-)}(\varphi)$ is given by 
(\ref{Fmdef}), which is a positive function. 

The second technical condition is also satisfied, as follows.
Using the asymptotics of the hypergeometric function 
\begin{equation}
{}_2F_1\left(\frac32, \beta; \frac52; 1 - \frac{1}{\varphi}\right) = 
\frac{\Gamma(\frac52) \Gamma(1-\beta)}{\Gamma(\frac52-\beta)} + O(\varphi^{-1})\,,
\qquad 
\text{as $\varphi\to \infty$}.
\end{equation}
we get that
\begin{equation}
[\mathcal{G}^{(-)}(\varphi)]^2 \sim \frac{(\varphi-1)^3}{\varphi^{2\beta}}\,,
\qquad 
\text{as $\varphi\to \infty$}.
\end{equation}
This has indeed superlinear growth provided that $\beta < 1$.

(ii) follows from Lemma 30 in \cite{ShortMatAsian}. This requires the following
two technical conditions: 
$\mathcal{G}^{(+)}(\chi)$ is a decreasing function, and 
the infimum in (\ref{IK2alt}) is not reached at the lower boundary $\chi=0$. 
The first condition follows indeed from (\ref{Fpldef}), as the integral in 
this expression is positive. 

The second condition follows by noting that we have, for $\beta\geq \frac12$
\begin{eqnarray}
\lim_{\chi\to 0} \frac{d}{d\chi} \left( 
\frac{\frac12 [\mathcal{G}^{(+)}(\chi)]^2}{\frac{K}{S_0}-\chi} \right) 
= - \infty \,.
\end{eqnarray}
This is obtained by writing the derivative explicitly
\begin{eqnarray}\label{ddchi}
\frac{d}{d\chi} \left( 
\frac{\frac12 [\mathcal{G}^{(+)}(\chi)]^2}{\frac{K}{S_0}-\chi} \right) =
\frac{\mathcal{G}^{(+)}(\chi) \frac{d}{\chi}\mathcal{G}^{(+)}(\chi)}
{\frac{K}{S_0}-\chi} +
\frac12 
\frac{ [\mathcal{G}^{(+)}(\chi)]^2}{(\frac{K}{S_0}-\chi)^2} \,.
\end{eqnarray}

Furthermore, the functions appearing here have the $\chi \to 0$ 
limits, for $\beta \geq \frac12$,
\begin{equation}\label{chi0}
\mathcal{G}^{(+)}(\chi) = 1 + O\left(\chi^{\frac32-\beta}\right),
\qquad\text{as $\chi\rightarrow 0$}
\end{equation}
and
\begin{equation}\label{dchi0}
\frac{d}{d\chi} \mathcal{G}^{(+)}(\chi) = -\infty
\qquad\text{as $\chi\rightarrow 0$} \,.
\end{equation}

The relation (\ref{chi0}) follows from the $\chi \to 0$ asymptotics of the
hypergeometric function, which can be extracted from Equation (\ref{HypG2})
\begin{equation}
{}_2F_1\left(\frac32, \beta; \frac52; 1 - \frac{1}{\chi}\right) =
\frac{3}{3-2\beta} \chi^\beta +
\frac{\Gamma(\frac52) \Gamma(\beta-\frac32)}{\Gamma(\beta)}
\chi^{3/2}\,.
\end{equation}
The relation (\ref{dchi0}) is obtained from (\ref{Fpldef}) by noting that the
integral on the RHS is bounded from below as 
\begin{eqnarray}
&& \int_\chi^1 \frac{dz}{z^\beta \sqrt{z-\chi}} \geq 
\int_{\chi}^1 dz z^{-\frac12-\beta} 
=\frac{1}{\frac12-\beta} (1 - \chi^{\frac12-\beta}) \to +\infty\,, \quad
\chi\to 0_+\,.
\end{eqnarray}
In the last step we used $\beta>\frac12$. The conclusion holds also for
$\beta=\frac12$, using the relation
\begin{eqnarray}
\int_\chi^1 \frac{dz}{ \sqrt{z(z-\chi)}} = 2 \log (\sqrt{1-\chi}+1) -
\log\chi
\to \infty \,,\quad \chi\to 0_+\,.
\end{eqnarray}

This shows that the infimum in (\ref{IK2alt}) is not reached at the lower 
boundary $\chi=0$.
This justifies the application of Lemma 30 in \cite{ShortMatAsian}.

iii) The conclusion follows immediately from the 
result for the rate function $\mathcal{I}(K,S_0)$ given by Theorem
\ref{ThmOTMCEV} and the monotonicity
properties of $\mathcal{I}_K(K,S_0)$ proven above in (i) and (ii).
\end{proof}

\begin{proof}[Proof of Proposition~\ref{prop:11}]

We give here the proof for the large-strike asymptotics of the 
rate function $\mathcal{I}(K,S_0)$.

For this case we are interested in the $x \to \infty$ asymptotics of the
functions $a^{(-)}(x), b^{(-)}(x)$. For this purpose it is useful to 
transform the argument $z=1-\frac{1}{x}$ of the hypergeometric functions 
appearing in the expressions of these functions as 
\begin{equation}
z \to 1-z = \frac{1}{x}
\end{equation}
using the identity 15.3.6 in Abramowitz and Stegun \cite{AS}.
\begin{align*}
{}_2 F_1(a,b;c;z) &= \frac{\Gamma(c)\Gamma(c-a-b)}{\Gamma(c-a)\Gamma(c-b)}
{}_2F_1(a,b;a+b-c+1;1-z) 
\\
&\qquad
+ (1-z)^{c-a-b} \frac{\Gamma(c)\Gamma(a+b-c)}{\Gamma(a)\Gamma(b)}
{}_2F_1(c-a,c-b; c-a-b+1;1-z).
\end{align*}

We get, for $\beta\in [\frac{1}{2},1)$, 
\begin{eqnarray}
{}_2 F_1\left(\beta,\frac12;\frac32;1-\frac{1}{x}\right) &=& 
\frac{\Gamma(\frac32)\Gamma(1-\beta)}{\Gamma(\frac32-\beta)} + O(x^{\beta-1})\,, \\
{}_2 F_1\left(\beta,\frac32;\frac52;1-\frac{1}{x}\right) &=& 
\frac{\Gamma(\frac52)\Gamma(1-\beta)}{\Gamma(\frac52-\beta)} + O(x^{\beta-1})\,,
\end{eqnarray}
as $x\to \infty$.

The solution of the equation (\ref{ratecase2}) for $x$ for $K/S_0 \gg 1$
is 
\begin{eqnarray}
x= \frac{3-2\beta}{2(1-\beta)} \left(\frac{K}{S_0}\right) + O(K/S_0)\,.
\end{eqnarray}
Substituting $x$ into the expression for the rate function of 
Proposition~\ref{VarProp} we obtain the large-strike
asymptotics of $\mathcal{I}(K,S_0)$ given in Proposition~\ref{prop:11}.
\end{proof}

\begin{proof}[Proof of Proposition~\ref{prop:12}]

We give here the proof for the small-strike asymptotics of the 
rate function $\mathcal{I}(K,S_0)$.

We require the $x \to 0_+$ asymptotics for $a^{(+)}(x), b^{(+)}(x)$.
This is obtained by changing the $z=1-\frac{1}{x}$ argument of the 
hypergeometric functions appearing in the expressions for these functions as
\begin{equation}
z \to \frac{1}{z-1} = -x,
\end{equation}
using the identity 15.3.8 in Abramowitz and Stegun \cite{AS}
\begin{align}
{}_2 F_1(a,b;c;z) &= (1-z)^{-a} \frac{\Gamma(c)\Gamma(b-a)}{\Gamma(b)\Gamma(c-a)}
{}_2F_1\left(a,c-b;a-b+1;\frac{1}{z-1}\right) 
\\
&\qquad
+ (1-z)^{-b} \frac{\Gamma(c)\Gamma(a-b)}{\Gamma(a)\Gamma(c-b)}
{}_2F_1\left(b,c-a; b-a+1;\frac{1}{z-1}\right).\nonumber
\end{align}
This can be used to find the asymptotics for $x\to 0_+$, 
together with the small-$x$ asymptotics
\begin{equation}
{}_2F_1(a,c-b;a-b+1;x) = 1 + O(x).
\end{equation}

We get
\begin{align}\label{HypG1}
{}_2F_1\left( \beta, \frac12; \frac32; 1 - \frac{1}{x} \right) 
&=
x^\beta
\frac{\Gamma(3/2) \Gamma(1/2-\beta)}{\Gamma(1/2)\Gamma(3/2-\beta)}
(1 + O(x))
\\
&
\qquad+x^{\frac12}
\frac{\Gamma(3/2) \Gamma(\beta-1/2)}{\Gamma(\beta)} 
(1 + O(x))\nonumber 
\\
&=x^{\beta}
\frac{1}{1-2\beta} (1 + O(x))
+ x^{\frac12}
\frac{\Gamma(3/2) \Gamma(\beta-1/2)}{\Gamma(\beta)} (1 + O(x)),
\nonumber
\end{align}
and
\begin{align}\label{HypG2}
{}_2F_1\left( \beta, \frac32; \frac52; 1 - \frac{1}{x} \right) 
&=
x^{\beta} 
\frac{\Gamma(5/2) \Gamma(3/2-\beta)}{\Gamma(3/2)\Gamma(5/2-\beta)}
(1 + O(x))
\\
&
\qquad+ x^{\frac32}
\frac{\Gamma(5/2) \Gamma(\beta-3/2)}{\Gamma(\beta)} 
(1 + O(x))
\nonumber 
\\
&=x^{\beta}
\frac{3}{3-2\beta} (1 + O(x))
+ x^{\frac32}
\frac{\Gamma(5/2) \Gamma(\beta-3/2)}{\Gamma(\beta)} 
(1 + O(x)).\nonumber
\end{align}
For $\frac12 < \beta<1$, the dominant term in these
expansions as $x\to 0_+$ is the second
term in (\ref{HypG1}), and the first term in (\ref{HypG2}).

The equation for $x$ as $K \to 0$ becomes approximatively
\begin{eqnarray}
\frac{K}{S_0} = x^{\beta-\frac12} 
\frac{\Gamma(\beta)}{\sqrt{\pi}(\frac32-\beta)\Gamma(\beta-\frac12)} + O(x)\,.
\end{eqnarray}
Substituting $x$ into the expression for the rate function of 
Proposition~\ref{VarProp} 
we obtain the small-strike
asymptotics of $\mathcal{I}(K,S_0)$ given in Proposition~\ref{prop:12}.
\end{proof}


\subsection{Proof of the results in Section~\ref{Sec:4}}

\begin{proof}[Proof of Theorem~\ref{ThmFloatingStrikeSqrt2}]
For any $\theta\in\mathbb{R}$, 
$\mathbb{E}[e^{\frac{\theta}{T^{2}}
\int_{0}^{t}S_{s}ds-\frac{\theta\kappa}{T}S_{T}}|S_{0}]
=e^{A(T;\frac{\theta}{T^{2}},-\frac{\theta\kappa}{T})S_{0}}$, 
where $A(t;\theta;\phi)$ satisfies the ODE:
\begin{equation}
A'(t;\theta,\phi)=(r-q)A(t;\theta,\phi)+\frac{1}{2}\sigma^{2}A(t;\theta,\phi)^{2}+\theta,
\end{equation}
with $A(0;\theta,\phi)=\phi$. 

For $\theta>0$,
\begin{align}
A(t;\theta,\phi)
&=\frac{\sqrt{2\sigma^{2}\theta-(r-q)^{2}}}{\sigma^{2}}
\tan\left[
\frac{\sqrt{2\sigma^{2}\theta-(r-q)^{2}}}{2}t
+\tan^{-1}\left(\frac{r-q+\sigma^{2}\phi}{\sqrt{2\sigma^{2}\theta-(r-q)^{2}}}\right)\right]
\\
&\qquad\qquad\qquad\qquad\qquad\qquad\qquad\qquad\qquad
-\frac{r-q}{\sigma^{2}},
\nonumber
\end{align}
and for $\theta<0$, 
\begin{align}
A(t;\theta,\phi)
&=\frac{(\frac{r-q}{\sigma^{2}}-\sqrt{\frac{(r-q)^{2}}{\sigma^{4}}-\frac{2\theta}{\sigma^{2}}}+\phi)
(\frac{r-q}{\sigma^{2}}+\sqrt{\frac{(r-q)^{2}}{\sigma^{4}}-\frac{2\theta}{\sigma^{2}}})
e^{t\sqrt{(r-q)^{2}-2\theta\sigma^{2}}}
}{
(\frac{r-q}{\sigma^{2}}+\sqrt{\frac{(r-q)^{2}}{\sigma^{4}}-\frac{2\theta}{\sigma^{2}}}+\phi)
-e^{t\sqrt{(r-q)^{2}-2\theta\sigma^{2}}}(\frac{r-q}{\sigma^{2}}-\sqrt{\frac{(r-q)^{2}}{\sigma^{4}}-\frac{2\theta}{\sigma^{2}}}+\phi)}
\\
&\qquad
-\frac{(\frac{r-q}{\sigma^{2}}-\sqrt{\frac{(r-q)^{2}}{\sigma^{4}}-\frac{2\theta}{\sigma^{2}}})
(\frac{r-q}{\sigma^{2}}+\sqrt{\frac{(r-q)^{2}}{\sigma^{4}}-\frac{2\theta}{\sigma^{2}}}+\phi)}
{
(\frac{r-q}{\sigma^{2}}+\sqrt{\frac{(r-q)^{2}}{\sigma^{4}}-\frac{2\theta}{\sigma^{2}}}+\phi)
-e^{t\sqrt{(r-q)^{2}-2\theta\sigma^{2}}}(\frac{r-q}{\sigma^{2}}-\sqrt{\frac{(r-q)^{2}}{\sigma^{4}}-\frac{2\theta}{\sigma^{2}}}+\phi)}\,.
\nonumber
\end{align}

For $0\leq\theta<\theta_{c}$,
\begin{equation}
\lim_{T\rightarrow 0}TA\left(T;\frac{\theta}{T^{2}},\frac{-\kappa\theta}{T}\right)=\sqrt{\frac{2\theta}{\sigma^{2}}}
\tan\left(\sqrt{\frac{\sigma^{2}\theta}{2}}+\tan^{-1}\left(-\sigma\kappa\sqrt{\frac{\theta}{2}}\right)\right),
\end{equation}
and this limit is $\infty$ if $\theta\geq\theta_{c}$, where
$\theta_{c}$ is the unique positive solution to the equation:
\begin{equation}\label{unique:sol}
\sqrt{\frac{\sigma^{2}\theta_{c}}{2}}+\tan^{-1}\left(-\sigma\kappa\sqrt{\frac{\theta_{c}}{2}}\right)=\frac{\pi}{2}.
\end{equation}
To see that \eqref{unique:sol} has a unique positive solution, let us define:
\begin{equation}
F(x):=\sqrt{\frac{\sigma^{2}}{2}}x+\tan^{-1}\left(-\sigma\kappa\frac{1}{\sqrt{2}}x\right)-\frac{\pi}{2}.
\end{equation}
Then, $F(0)=-\frac{\pi}{2}$ and $F(\infty)=\infty$. On the other hand, we can compute that
\begin{equation}
F'(x)=\sqrt{\frac{\sigma^{2}}{2}}-\frac{\sigma\kappa}{\sqrt{2}}\frac{1}{\frac{1}{2}\sigma^{2}\kappa^{2}x^{2}+1},
\qquad
F''(x)=\frac{\sigma\kappa}{\sqrt{2}}\frac{\sigma^{2}\kappa^{2}x}{(\frac{1}{2}\sigma^{2}\kappa^{2}x^{2}+1)^{2}}.
\end{equation}
Since $F''(x)>0$ for any $x>0$, and $F(0)=-\frac{\pi}{2}<0$ and $F(\infty)=\infty$, 
it follows that $F(x)=0$ has a unique positive solution.

For $\theta<0$,
\begin{align}
\lim_{T\rightarrow 0}TA\left(T;\frac{\theta}{T^{2}},\frac{-\kappa\theta}{T}\right)
&=-\frac{\sqrt{-2\theta}}{\sigma}\frac{\frac{\sqrt{-2\theta}}{\sigma}(e^{\sigma\sqrt{-2\theta}}-1)+\theta\kappa
(1+e^{\sigma\sqrt{-2\theta}})}
{\frac{\sqrt{-2\theta}}{\sigma}(1+e^{\sigma\sqrt{-2\theta}})-\theta\kappa(1-e^{\sigma\sqrt{-2\theta}})}
\\
&=-\frac{\sqrt{-2\theta}}{\sigma}\frac{\frac{\frac{\sqrt{-2\theta}}{\sigma}+\theta\kappa}{\frac{\sqrt{-2\theta}}{\sigma}-\theta\kappa}
e^{\sigma\sqrt{-2\theta}}-1}{\frac{\frac{\sqrt{-2\theta}}{\sigma}+\theta\kappa}{\frac{\sqrt{-2\theta}}{\sigma}-\theta\kappa}
e^{\sigma\sqrt{-2\theta}}+1}
\nonumber
\\
&=-\frac{\sqrt{-2\theta}}{\sigma}
\tanh\left(\frac{\sigma}{2}\sqrt{-2\theta}+\tanh^{-1}\left(-\sigma\kappa\sqrt{\frac{-\theta}{2}}\right)\right).
\nonumber
\end{align}
Therefore,
\begin{align}
\Lambda(\theta)&:=\lim_{T\rightarrow 0}T\log\mathbb{E}\left[e^{\frac{\theta}{T^{2}}\int_{0}^{T}S_{t}dt-\frac{\theta}{T}\kappa S_{T}}\right]
\\
&=
\begin{cases}
\frac{\sqrt{2\theta}}{\sigma}
\tan\left(\frac{\sigma}{2}\sqrt{2\theta}+\tan^{-1}\left(-\sigma\kappa\sqrt{\frac{\theta}{2}}\right)\right)S_{0} &\text{if $0\leq\theta<\theta_{c}$}
\\
-\frac{\sqrt{-2\theta}}{\sigma}
\tanh\left(\frac{\sigma}{2}\sqrt{-2\theta}+\tanh^{-1}\left(-\sigma\kappa\sqrt{\frac{-\theta}{2}}\right)\right)S_{0} 
&\text{if $\theta\leq 0$}
\\
+\infty &\text{otherwise}
\end{cases}.
\nonumber
\end{align}
It is easy to show that $\Lambda_{f}(\theta)$ is differentiable for any $\theta<\theta_{c}$
and $\Lambda'_{f}(\theta)\rightarrow\infty$ as $\theta\uparrow\theta_{c}$.
Hence, 
$\mathbb{P}\left(\frac{1}{T}\int_{0}^{T}S_{t}dt-\kappa S_{T}\in\cdot\right)$
satisfies a large deviation principle with the rate function 
$\mathcal{I}_f(\kappa, S_0)$ given in (\ref{IfSqrt})
by applying the G\"{a}rtner-Ellis theorem, see Theorem~\ref{GEThm} in the Appendix.
\end{proof}


\section*{Acknowledgements}

Lingjiong Zhu acknowledges the support from NSF Grant DMS-1613164.

\end{document}